\documentclass[11pt]{article}

\usepackage{amsmath}
\usepackage{hyperref}
\usepackage{amsfonts}
\usepackage{amssymb}
\usepackage{graphicx}
\usepackage{multicol,float}
\usepackage{color}
\usepackage{multirow}
\usepackage{amsthm}
\usepackage{dsfont}

\providecommand{\U}[1]{\protect\rule{.1in}{.1in}}

\newtheorem{theorem}{Theorem}[section]

\newtheorem{assumption}{Assumption}[section]

\newtheorem{corollary}{Corollary}[section]

\newtheorem{definition}{Definition}[section]
\newtheorem{example}{Example}[section]

\newtheorem{lemma}{Lemma}[section]

\newtheorem{proposition}{Proposition}[section]
\newtheorem{remark}{Remark}[section]

\numberwithin{equation}{section}

\newcommand \E {\mathbb{E}}
\newcommand \R {\mathbb{R}}
\renewcommand \P {\mathbb{P}}
\newcommand \N {\mathbb{N}}

\newcommand \mX {\mathcal{X}}
\newcommand \mi {\mathcal{I}}
\newcommand \mic {\mi_c}
\newcommand \mD {\mathcal{D}}
\newcommand \mM {\mathcal{M}}
\newcommand \mC {\mathcal{C}}

\newcommand{\id}{\mathds{1}}

\newcommand \Sy {S_Y}
\newcommand \Sx {S_X}
\newcommand \Fx {F_X}
\newcommand \fx {f_X}

\newcommand \eps {\varepsilon}
\newcommand \tet {\theta}
\newcommand \gam {\gamma}

\newcommand \alp {\alpha}

\newcommand \la {\lambda}
\newcommand \kap {\kappa}
\newcommand \Ups {\Upsilon}

\newcommand \Is {I^*}
\newcommand \Rs {R^*}
\newcommand \pis {\pi^*}
\newcommand \Ieps {I_\eps}
\newcommand \Reps {R_\eps}
\newcommand \pieps {\pi_\eps}

\newcommand \Rd {R_d}
\newcommand \Id {I_d}
\newcommand \pid {\pi_d}
\newcommand \ds {d^*}
\newcommand \ms {m^*}

\newcommand \tk {\tilde k}
\newcommand \tb {\tilde b}

\allowdisplaybreaks

\usepackage{color}

\begin{document}

\title{Optimal Insurance to Maximize RDEU \break Under a Distortion-Deviation Premium Principle}

\author{Xiaoqing Liang
\thanks{Department of Statistics, School of Sciences, Hebei University of Technology, Tianjin 300401, P.\ R.\ China, liangxiaoqing115@hotmail.com. X.\ Liang thanks the National Natural Science Foundation of China (11701139, 11571189) and the Natural Science Foundation of Hebei Province (A2018202057) for financial support.}
\and Ruodu Wang
\thanks{Department of Statistics and Actuarial Science, University of Waterloo, Waterloo, Ontario N2L 3G1, Canada, wang@uwaterloo.ca.  R.\ Wang acknowledges financial support from the Natural Sciences and Engineering Research Council of Canada (NSERC, RGPIN-2018-03823, RGPAS-2018-522590) and from a Center of Actuarial Excellence Research Grant from the Society of Actuaries.
}
\and Virginia R. Young
\thanks{Corresponding author. Department of Mathematics, University of Michigan, Ann Arbor, Michigan, 48109, vryoung@umich.edu. V.\ R.\ Young thanks the Cecil J. and Ethel M. Nesbitt Professorship of Actuarial Mathematics for financial support.}
}

\date{\today}

\maketitle

\begin{abstract}
In this paper, we study an optimal insurance problem for a risk-averse individual who seeks to maximize the rank-dependent expected utility (RDEU) of her terminal wealth, and insurance is  priced via a general distortion-deviation premium principle.  We prove necessary and sufficient conditions satisfied by the optimal solution and consider three orders between the distortion functions for the buyer and the seller to further determine the optimal indemnity.  Finally, we analyze examples under three distortion-deviation premium principles to explore the specific conditions under which no insurance or deductible insurance is optimal.

\medskip


\medskip


\medskip


\end{abstract}


\section{Introduction}

Optimal insurance design is an important topic in insurance economics and has been studied for decades.  
To reduce the non-hedgeable claim risk $X$, the buyer of insurance can transfer part of the claim $I(X)$ (called the {\it indemnity}) to the insurer by paying a premium to the insurer, which is a functional of $I(X)$.  The optimal insurance design problem is to determine the optimal indemnity $\Is(X)$ that maximizes some measure of the buyer's well-being.

The optimal insurance design problem was first studied by Arrow \cite{A1963, A1971}, in which the insurer is assumed to be risk neutral and the buyer is assumed to be risk averse, as characterized, for instance, by having a concave utility.  The buyer finds the optimal indemnity to maximize her expected utility (EU) when premium equals an increasing function of the indemnity's expectation, and the optimal indemnity is deductible insurance.  See, for instance, the discussion in Promislow and Young \cite{PY2005}, Wang \cite{W1995, W1996}, and Chi and Zhou \cite{CZ2017}, among many other papers.

The literature on optimal insurance design with EU preferences is large.  However, some research reveals that EU theory fails to explain numerous common phenomena; therefore, Quiggin \cite{Q1982} proposed rank-dependent expected-utility theory (RDEU), which overcomes some drawbacks of EU theory.  Recently, a number of researchers considered RDEU preferences in an insurance setting.  For example, Bernard et al.\ \cite{BHY2015} solved an optimal insurance design problem under RDEU theory with a concave utility function and an inverse-S shaped probability distortion function. By applying the technique of quantile formulation, they solve the problem and find that the optimal contract is the classical deductible one for both large and small losses. But, their results failed to exclude the situations that the buyer might misreport actual losses. In order to rule out this severe moral hazard problem, Xu et al.\ \cite{XZZ2019} revisited the problem of Bernard et al.\ \cite{BHY2015} by adding an incentive-compatibility constraint, which requires the indemnity and retention functions to be increasing with respect to the loss.  Ghossoub \cite{G2019b} extended Bernard et al.\ \cite{BHY2015} and Xu et al.\ \cite{XZZ2019} by including a cost for state verification.

In the actuarial literature, researchers also use distorted probabilities when computing premiums.  Young \cite{Y1999} studied an expected utility maximization problem under Wang's premium principle (that is, a Choquet integral with a non-decreasing concave distortion function).\footnote{We emend Theorem 3.6 in Young \cite{Y1999} by the work in this paper, specifically, in Corollary \ref{cor:Young1999}.}  Wang et al.\ \cite{WWW2020b, WWW2020a} explored a class of risk functionals, called {\it signed Choquet integrals}.  As compared with the more usual {\it increasing} Choquet integrals, a signed one is not necessarily monotone.  They proved various properties of signed Choquet integrals, and found that many useful mathematical results for traditional risk functionals continue to hold for the signed Choquet integrals, that is, they do not depend on the property of monotonicity.  We note that some popular premium principles, such as the mean-variance and the mean-standard deviation principles, are indeed not monotone.


In this paper, we study an optimal insurance problem for an individual who seeks to maximize the rank-dependent expected utility (RDEU) of her terminal wealth.  Motivated by Wang et al.\ \cite{WWW2020b, WWW2020a}, we assume that insurance is priced via a non-monotone premium principle.  In Theorem \ref{thm:optins}, we provide necessary and sufficient conditions satisfied by the optimal solution. Furthermore, we define order relations between the distortions of the buyer and the seller based on those from the classic literature of stochastic orders, namely, first-order stochastic dominance, hazard rate order, and likelihood ratio order.  We determine properties of the optimal insurance based on these orderings of the distortion functions.  Finally, we revisit examples from Young \cite{Y1999} and extend them to explore the specific conditions under which no insurance or deductible insurance is optimal.

The approach in our paper is closely related to Chi and Zhuang \cite{CZ2020}, who focused on heterogeneous beliefs about the loss distributions in an optimal insurance design problem. We extend their model to consider a general premium principle without monotonicity under a risk-averse rank-dependent expected utility.  Our premium principle not only includes the traditional expectation measure, expected shortfall, but also other non-monotone deviation measures, such as the Gini deviation and the mean-median deviation, which can be seen as an analog of the mean-standard deviation premium principle.\footnote{See Example 3 in Wang et al.\ \cite{WWW2020a}, who proved that the standard deviation can be shown to equal a supremum over a collection of signed Choquet integrals.}  Our work differs from Chi and Zhuang \cite{CZ2020} in that we use (1) a more general premium principle, that is, we introduce the distortion-deviation premium principle; (2) a more general measure of risk aversion of the buyer, as measured by risk-averse RDEU; and, (3) a broader class of allowable indemnity policies.  That said, the proofs of the results in Section \ref{sec:optins}  for our more general model closely follow those of corresponding results in Chi and Zhuang \cite{CZ2020}.

The remainder of this paper is organized as follows. In Section \ref{sec:model}, we formulate the model of our problem.  In Section \ref{sec:ddpp}, we introduce the distortion-deviation premium principle, and in Section \ref{sec:buyer}, we propose the buyer's problem under rank-dependent expected utility.  By invoking the Comonotonic Improvement Theorem (see Landsberger and Meilijson \cite{LM1994}), we show that we can restrict our attention to indemnities $I(X)$ such that $I(X)$ and $X - I(X)$ are comonotonic, and we prove the existence of an optimal indemnity $\Is$.  We also prove necessary and sufficient conditions for the uniqueness of $\Is$.   Section \ref{sec:optins} contains our main results, beginning with Theorem \ref{thm:optins}, which characterizes an optimal indemnity.  By considering three orders between the distortion functions of the buyer and the seller, we further determine the form of an optimal indemnity.  Section \ref{sec:Y1999} contains specific examples to illustrate our results, and Section \ref{sec:conclu} ends the paper.

\section{Model}\label{sec:model}

Throughout the paper, let $(\Omega, \mathcal{F}, \P)$ be a probability space.  An individual faces a random loss $X$, with $X \ge 0$ almost surely on $\Omega$ and with $\E X < \infty$, and she wishes to indemnify her loss via insurance.  We assume that she chooses an indemnity from the following set of functions on $\R^+$:
\begin{equation}\label{eq:mi}
\mi = \big\{ I: I \hbox{ maps } \R^+ \hbox{ to } \R^+, \, 0 \le I(x) \le x \hbox{ for all } x \in \R^+ \big\},
\end{equation}
Consider the $L^1$ norm on $\mi$; specifically, if $\Fx$ is the cumulative distribution function of $X$, then the $L^1(\Fx)$ norm of $I \in \mi$ equals
\begin{equation}\label{eq:norm}
|| I ||_1 = \int_0^\infty \big| I(x) \big| \, d\Fx(x).
\end{equation}
$(\mi, || \cdot ||_1)$, modulo the set of null functions, those for which their $L^1(\Fx)$ norm equals $0$, or equivalently, those that are equal to $0$ almost everywhere with respect to the measure induced by $\Fx$, forms a closed subset of the complete, normed vector space $L^1(\Fx)$.  Note that $|| I ||_1 \le \E X < \infty$ for all $I \in \mi$.  

In the next section, we describe the principle that the insurer uses to compute its premium; in the section after that, we describe the optimization problem faced by the buyer of insurance.

\subsection{Distortion-deviation premium principle}\label{sec:ddpp}

We assume that the insurer computes premium for insurance according to the so-called {\it distortion-deviation} premium principle.  We introduce this premium principle as an analog of the mean-standard deviation premium principle, which computes premiums as follows:
\begin{equation}\label{eq:stddevpp}
(1 + \tet) \E Y + \alp \sqrt{\mathrm{Var} Y},
\end{equation}
for some constants $\tet \ge 0$ and $\alp \ge 0$, and for any non-negative random variable $Y$ with finite second moment.  In the following, we describe the two components of this new distortion-deviation premium principle; each component creates an analogy of $(1 + \tet) \E Y$ and $\alp \sqrt{\mathrm{Var} Y}$, in turn.\footnote{In research related to variance or standard deviation premium principles, Gajek and Zagrodny \cite{GZ2000} found the optimal insurance under the standard premium principle to minimize the variance of retained losses; they showed that the optimal indemnity is deductible insurance with a constant proportion for losses over the deductible.  Kaluszka \cite{K2001, K2004} derived the optimal reinsurance indemnity to minimize the variance of retained losses, under more general variance-related premium principles, and obtained the same form of indemnity as did Gajek and Zagrodny \cite{GZ2000}.  When minimizing the probability of insurer ruin or drawdown, Liang et al.\ \cite{LLY2020} and Azcue et al.\ \cite{ALMY2021} showed that the same form of per-loss reinsurance is optimal under the mean-variance premium principle and under the diffusion approximation to the classical Cram\'er-Lundberg risk model.} 

Wang \cite{W1996} introduced the well-known distortion premium principle as a generalization of his work with the proportional hazards transform in \cite{W1995}.  Let $\mD$ denote the collection of continuous, concave {\it distortions}; specifically,
\[
\mD = \big\{j: j \hbox{ maps } [0, 1] \hbox{ to } \R^+, \, j \hbox{ continuous and concave}, \, j(0) = 0 \big\}.
\]
Note that we do not require $j \in \mD$ to be monotone.  Then, for a distortion $j \in \mD$, define $\rho_j$ on the set of random variables $\mX = \{ X: \P\hbox{-ess inf } X > - \infty \big\}$ by
\begin{equation}\label{eq:dist}
\rho_j(Y) = \int_{-\infty}^0 \big( j(\Sy(t)) - j(1) \big) dt + \int_0^\infty j(\Sy(t)) \, dt,
\end{equation}
in which $\Sy$ is the survival function of $Y$, that is, $\Sy(t) = \P(Y > t)$.\footnote{We assume that distortions are concave from the outset so that $\rho_j$ preserves convex order, which is equivalent to second-order stochastic dominance (SSD) with equal means; see, for example, Wang and Young \cite{WY1998}.  Specifically, if $Y_1 \preceq_{cx} Y_2$, then $\rho_j(Y_1) \le \rho_j(Y_2)$.  The consistency between the concavity of $j$ and $\rho_j$ preserving convex order does not require $j \in \mD$ to be monotone; see Theorem 3 in Wang et al.\ \cite{WWW2020b}.}  In writing these integrals, we assume they are finite.

The risk measure $\rho_g(Y)$ generalizes $(1 + \tet)\E Y$.  Indeed, if $g$ is given by $g(p) = (1 + \tet)p$ for $p \in [0, 1]$, then $\rho_g(Y) = (1 + \tet)\E Y$.  We use $\rho_g(Y)$ to generalize the term $(1 + \tet)\E Y$ in our premium principle, and we assume that $g \in \mD$ is non-decreasing because $\E Y$ is non-decreasing in $Y$. One usually assumes that $g(1) \ge 1$, but we will allow $g(1) < 1$ for more generality.

Rockafellar et al.\ \cite{RUZ2006} introduced the idea of deviation measures, and Wang et al.\ \cite{WWW2020a} and Wang et al.\ \cite{WWW2020b} unified risk measures (such as the distortion premium principle) and deviation measures into a common framework called {\it distortion riskmetrics}.  To create an analog of $\alp \sqrt{\mathrm{Var} Y}$, which is itself a type of deviation measure, we consider $\rho_h(Y)$ for $h \in \mD$ with $h(0) = h(1)$, which necessarily equals $0$ because $h \in \mD$ implies $h(0) = 0$.  Note that $h(0) = h(1) = 0$ implies that probabilities that are certain (either $0$ or $1$) are not distorted by $h$ because it is impossible to deviate from a probability of $0$ or $1$.  We label such an $h$ a {\it deviation} distortion.  Recall that  $h$ is concave because it is in $\mD$.

We can obtain a symmetric deviation distortion $h$ by starting with any non-decreasing distortion $\tilde h \in \mD$ and defining $h$ by
\begin{equation}\label{eq:h}
h(p) = \tilde h(p) + \tilde h(1 - p) - \tilde h(1),
\end{equation}
for all $p \in [0, 1]$.  The distortion $h \in \mD$ defined by \eqref{eq:h} is symmetric with respect to $p = 1/2$, and $\rho_h(Y) = \rho_{\tilde h}(Y) + \rho_{\tilde h}(-Y)$.  Nevertheless, we wish to allow for non-symmetric deviation distortions $h$, so we do not assume the form in \eqref{eq:h}.

In the following example, we present two concave distortions $h \in \mD$ that are symmetric with respect to $p = 1/2$, which implies $h(0) = h(1)$.  Thus, these two distortions yield deviation measures via $h \mapsto \rho_h$.

\begin{example}\label{ex:1}{\rm
Suppose $h(p) = p - p^2$, which yields the {\it Gini deviation} measure defined by
\begin{equation}\label{eq:Gini}
\rho_h(Y) = \dfrac{1}{2} \, \E \big|Y - Y' \big|,
\end{equation}
in which $Y'$ is an i.i.d.\ copy of $Y$.  If we define $\tilde h$ by
\[
\tilde h(p) = \dfrac{3}{2} \, p - \dfrac{1}{2} \, p^2,
\]
then $\tilde h$ is increasing, concave with $\tilde h(1) = 1$, that is, $\tilde h$ is a traditional distortion, and $h$ and $\tilde h$ satisfy \eqref{eq:h}.

Another useful symmetric distortion is given by $h(p) = p \wedge (1 - p)$ for $p \in [0, 1]$.  This distortion yields the {\it mean-median deviation} measure defined by}
\begin{equation}\label{eq:mm}
\rho_h(Y) = \min_{y \in \R} \E \big|Y - y \big|.
\end{equation}
\end{example}

\bigskip

Based on the discussion preceding Example \ref{ex:1}, we define the {\it distortion-deviation} premium principle $\pi$ for $Y \in \mX$ as follows:
\begin{equation}\label{eq:ddpp}
\pi(Y) = \rho_g(Y) + \rho_h(Y),
\end{equation}
in which $g, h \in \mD$ are such that $g$ is non-decreasing and $h$ is a deviation distortion.  Throughout this paper, we assume that $X$, $g$, and $h$ are such that $\pi(I(X))$ is finite for all $I \in \mi$.

The representation of $\pi$ in \eqref{eq:ddpp} is not unique; however, there is a unique representation of $g + h$ as the sum of a linear function and a deviation distortion, as we prove in the following proposition.

\begin{proposition}\label{prop:canon}
For $\pi$ given in \eqref{eq:ddpp}, define $\tet > -1$ by
\begin{equation}\label{eq:tet}
\tet = g(1) - 1,
\end{equation}
and define $k \in \mD$ by
\begin{equation}\label{eq:k}
k(p) = g(p) + h(p) - (1 + \tet)p.
\end{equation}
Then, $k(0) = k(1) = 0$, and
\begin{equation}\label{eq:canon}
\pi(Y) = (1 + \tet)\E Y + \rho_k(Y),
\end{equation}
for all $Y \in \mX$.  Moreover, the representation of $g + h$ as the sum of a linear function and a deviation distortion is unique.
\end{proposition}

\begin{proof}
The line $(1 + \tet) p$ is a secant of the graph of $g + h$.  Because $g + h$ is concave, then it lies above this secant, which implies that $k$ defined in \eqref{eq:k} is non-negative.  It is easy to see that $k$ is concave, and the definition of $\tet$ in \eqref{eq:tet} implies $k(0) = k(1) = 0$, so $k$ is a deviation distortion.  It is also straightforward to show that the representation $g(p) + h(p) = (1 + \tet)p + k(p)$, with $k$ a deviation distortion, is unique.
\end{proof}

We call the representation of $\pi$ in \eqref{eq:canon} the {\it canonical} representation, and we use it in the remainder of this paper.\footnote{Throughout this paper, we will use $\rho_j$ as in \eqref{eq:dist} to represent a distortion risk measure or deviation measure, and we will use $\pi$ to specifically mean the distortion-deviation premium principle given in \eqref{eq:canon}.}  The canonical representation in \eqref{eq:canon} highlights the parallel between the distortion-deviation premium principle and the mean-standard deviation premium principle in \eqref{eq:stddevpp} because the first terms match (except that we naturally require $\tet > -1$) and because $\rho_k$ is a measure of deviation.  We remark that in case $h \equiv 0$, that is, in the absence of the deviation term in \eqref{eq:ddpp}, the distortion premium itself can be decomposed into an expected-value premium and a non-zero distortion deviation according to Proposition \ref{prop:canon}.

The following lemma lists some useful properties of $\pi$.  We present the lemma without proof because the proof is standard in the literature on distortions; see Wang et al.\ \cite{WWW2020b} or Wang et al.\ \cite{WWW2020a} for a proof of this lemma when the distortion is not necessarily monotone.

\begin{lemma}\label{lem:pi_prop}
The distortion-deviation premium principle given in \eqref{eq:canon} satisfies the following$:$
\begin{enumerate}
\item[$1.$]  For all $Y \in \mX$ and $c \in \R$, $\pi(Y + c) = \pi(Y) + (1 + \tet)c$.

\item[$2.$]  For all $Y \in \mX$ and $c \ge 0$, $\pi(cY) = c \pi(Y)$.

\item[$3.$]  For all $Y, Z \in \mX$ and $\la \in [0, 1]$, $\pi(\la Y + (1 - \la)Z) \le \la \pi(Y) + (1 - \la)\pi(Z)$.

\item[$4.$]  If $Y, Z \in \mX$ are such that $Y \preceq_{cx} Z$, then $\pi(Y) \le \pi(Z)$.

\item[$5.$]  If $Y, Z \in \mX$ are comonotonic, that is, if $(Y(\omega) - Y(\omega'))(Z(\omega) - Z(\omega')) \ge 0$ almost surely on $\Omega \times \Omega$, then $\pi(Y + Z) = \pi(Y) + \pi(Z)$.  \qed
\end{enumerate}
\end{lemma}

\subsection{Buyer's problem}\label{sec:buyer}

The individual, who faces the non-negative random loss $X$, chooses an indemnity from $\mi$ to maximize her rank-dependent expected utility of terminal wealth:
\begin{align}\label{eq:RDEU}
\rho_b\big( u(w - X + I(X) - \pi_I) \big) &= \int_{-\infty}^0 \big( b \circ \P \big(u(w - X + I(X) - \pi_I ) > t \big) - b(1) \big) dt \notag \\
&\quad + \int_0^\infty b \circ \P \big(u(w - X + I(X) - \pi_I ) > t \big) dt,
\end{align}
in which $\pi_I = \pi(I(X))$.  In \eqref{eq:RDEU}, $u \in \mC^2(\R)$ is a strictly increasing, concave utility function, $b$ is a continuously differentiable, strictly increasing, {\it convex} distortion function with $b(0) = 0$ and $b(1) = 1$, $w$ is the individual's initial wealth, and $\pi_I$ is computed according to the distortion-deviation premium principle in \eqref{eq:canon}.\footnote{If $X$ has finite support, then we will not need $u$ to be defined on {\it all} of $\R$.  Also, assuming $b(1) = 1$ is without loss of generality because, if we were to scale $b$ by any positive constant, then $\rho_b$ would scale by that same constant, which would not affect the optimality of a given indemnity.  Finally, because we assume the distortion $b$ is convex, we do not include inverse-S shaped distortions, as studied by Bernard et al.\ \cite{BHY2015}.}  One can think of $u$ and $b$ as measuring the buyer's risk aversion concerning wealth and uncertainty, respectively.\footnote{Putting our model in the format of Table 2 in Ghossoub and He \cite{GH2020}, the buyer of insurance (that is, Party A) has a concave utility $u$ and a concave probability weighting via the convex distortion $b$.  Furthermore the seller of insurance (that is, Party B) has a linear utility and a concave probability weighting via the concave distortion $(1 + \tet)p + k(p)$, but the seller's distortion is not necessarily monotone.}  Throughout this paper, we assume that $b$ and $u$ are such that $\rho_b\big( u(w - I - \pi_I) \big)$ is finite for all $I \in \mi$.

We prove the following lemma that will allow us to restrict our attention to indemnities $I$ such that $I(X)$ and $X - I(X)$ are comonotonic, that is, the set of functions defined by
\begin{equation}\label{eq:mic}
\mic = \big\{ I \in \mi: 0 \le I(x) - I(y) \le x - y, \hbox{ for all } 0 \le y \le x\big\}.
\end{equation}

\begin{lemma}\label{lem:comon}
Because $\rho_b\big( u(w - X + I(X) - \pi_I) \big)$ in \eqref{eq:RDEU} is decreasing in convex order,
\[
\sup_{I \in \mi} \rho_b\big( u(w - X + I(X) - \pi_I) \big) = \sup_{I \in \mic} \rho_b\big( u(w - X + I(X) - \pi_I) \big).
\]
\end{lemma}

\begin{proof}
By the Comonotonic Improvement Theorem (see, for example, Theorem 10.50 in R\"uschendorf \cite{R2013}), because $\E X < \infty$, for any allocation $(I(X), X - I(X))$, with $I \in \mi$, there exists a comonotonic allocation $(I_c(X), X - I_c(X))$, with $I_c \in \mic$, such that
\[
I_c(X) \preceq_{cx} I(X),
\]
and
\[
X - I_c(X) \preceq_{cx} X - I(X).
\]
Then, because $g$ and $h$ are concave, $\pi$ is increasing in convex order (see Property 4 in Lemma \ref{lem:pi_prop}), which implies
\[
\pi_{I_c} \le \pi_I,
\]
from which it follows (because $u$ is increasing)
\[
u(w - X + I(X) - \pi_{I_c}) \ge u(w - X + I(X) - \pi_I), \; \text{a.s.},
\]
and taking expectations with respect the distorted (non-additive) measure $b \circ \P$ as in \eqref{eq:RDEU} gives us
\[
\rho_b\big( u(w - X + I(X) - \pi_{I_c})\big) \ge \rho_b\big(u(w - X + I(X) - \pi_I) \big).
\]
Moreover, because $u$ is concave and $b$ is convex, $\rho_b (u(\cdot))$ is decreasing in convex order (see, for example, Chew et al.\ \cite{CKS1987}), which implies
\[
\rho_b \big(u(w - X + I_c(X) - \pi_{I_c}) \big) \ge \rho_b\big( u(w - X + I(X) - \pi_{I_c}) \big).
\]
By combining the previous two inequalities, we obtain
\[
\rho_b \big(u(w - X + I_c(X) - \pi_{I_c}) \big) \ge \rho_b \big(u(w - X + I(X) - \pi_I) \big),
\]
and the lemma follows.
\end{proof}

\begin{remark}\label{rem:two}
There are two approaches to modeling in the optimal-insurance problem: $(1)$ Allow risk preferences, as modeled by $(1 + \tet)p + k(p)$ for the seller of insurance and by $b$ and $u$ for the buyer of insurance, to be as general as possible, but restrict the set of possible indemnities in such a way so that one gets meaningful results.  $(2)$ Restrict risk preferences $($in our case, by requiring that $b$ be convex$)$, but allow the set of possible indemnities to be rather general, again, in such a way so that one gets meaningful results.  We chose to do the latter, and we use the convexity of $b$ in the proof of Lemma {\rm \ref{lem:comon}} to prove that we can restrict our attention to indemnities in $\mic$.  By contrast, we could have restricted our attention to indemnities in $\mic$ from the outset and allowed $b$ to be strictly increasing, but not necessarily convex.  Indeed, in the following, all our results hold if we had restricted indemnities to lie in $\mic$ and had assumed $b$ is strictly increasing, instead of strictly increasing and convex, unless otherwise stated in the proposition.   \qed
\end{remark}

Lemma \ref{lem:comon} implies that the individual's problem is equivalent to the following:
\begin{equation}\label{eq:max}
\begin{cases}
\sup \limits_{I \in \mic} \rho_b \big(u(w - X + I(X) - \pi_I) \big), \\
\hbox{in which } \pi_I = \pi(I(X)) = (1 + \tet)\E I(X) + \rho_k(I(X)),
\end{cases}
\end{equation}
with $\tet > -1$ and $k \in \mD$ satisfying $k(0) = k(1) = 0$.  
In the next proposition, we prove that \eqref{eq:max} has a solution in $\mic$, and we need the following lemma in the proof of that proposition.

\begin{lemma}\label{lem:mic_compact}
$(\mic, || \cdot ||_1)$ is a compact subset of the complete, normed vector space $L^1(\Fx)$, in which $|| \cdot ||_1$ is given in \eqref{eq:norm}.
\end{lemma}

\begin{proof}
We prove this lemma via information from the survey paper by Hanche-Olsen and Holden \cite{HH2010}.  First, a subset of a metric space is compact if and only if it is complete and totally bounded.  $\mic$ is a closed subset of the complete space $L^1(\Fx)$; therefore, $\mic$ is complete.  To prove that $\mic$ is totally bounded, we rely on Lemma 1 of Hanche-Olsen and Holden \cite{HH2010}.

First, because for any $I \in \mic$, $|| \cdot ||_1 \le \E X < \infty$; thus, $\mic$ is bounded by $\E X$.  For $\eps > 0$,  because $\E X < \infty$, there exists $M > 0$ such that $\int_M^\infty x \, d\Fx(x) < \eps$.  Then, because $0 \le I(x) \le x$ for all $x \ge 0$ and for all $I \in \mic$, it follows that $\int_M^\infty I(x) \, d\Fx(x) < \eps$ for all $I \in \mic$.

Next, choose $N = \big\lfloor \frac{M}{\eps} \big\rfloor + 1$, and define $Q = \bigcup^N_{i=1} Q_i $, in which $Q_i = \big((i-1)\eps, \, i \eps \big]$. We see that the $Q_i$, for $i = 1, \dots, N$ are mutually non-overlapping and of equal length $\eps$.  Moreover, $Q \supset (0, M]$ and for any $x, y \in Q_i$, we have $|x-y|< \eps$. Let 
$P: \mic \rightarrow \R^+$ denote the projection mapping of $L^1(\Fx)$ onto the linear span of the characteristic functions of $Q_i$ given by
\begin{equation*}
\big(P(I)\big)(x) =
\begin{cases}
\dfrac{1}{\eps} \displaystyle {\int^{i \eps}_{(i-1)\eps} I(y) dy}, &\quad x\in Q_i, \hbox{ for some } i=1, \dots, N, \\
0, &\quad \rm{otherwise}.
\end{cases} 
\end{equation*}
Then, we have for $I \in \mic$,
\begin{align*}
||I - P(I) ||_1 & = \int^{\infty}_0 \big|I(x) - \big(P(I)\big)(x) \big| d\Fx(x) = \int^{N \eps}_0 \big|I(x) - \big(P(I)\big)(x) \big| d\Fx(x) + \int^{\infty}_{N\eps} I(x) d\Fx(x) \\
& \le \sum^N_{i=1} \int^{i \eps}_{(i-1)\eps} \left|I(x) - \frac{1}{\eps} \int^{i \eps}_{(i-1)\eps} I(y) dy \right| d\Fx(x) + \int^{\infty}_{M} I(x) d\Fx(x) \\
& \le \sum^N_{i=1} \int^{i \eps}_{(i-1)\eps}\frac{1}{\eps} \int^{i \eps}_{(i-1)\eps} | I(x) -   I(y)| \, dy \, d\Fx(x) + \eps \\
& \le \eps \big(\Fx(N\eps) + 1 \big) \le 2 \eps,
\end{align*}
in which the third inequality follows from $| I(x) -   I(y)| \le |x - y| < \eps$, for any $x, y \in Q_i$.

Moreover, if $I_1, I_2 \in \mic$ and $||P(I_1) - P(I_2)||_1 < \eps$, then $||I_1 - I_2||_1 < 5 \eps$. Indeed, by using the triangle inequality in $L^1$-norm, we obtain
\begin{align*}
|| I_1 - I_2 ||_1  &= \big|\big| \big(P(I_1) - I_1\big) - \big(P(I_2) - I_2 \big) - \big(P(I_1) - P(I_2)\big) \big|\big|_1 \\
& \le \big|\big| P(I_1) - I_1\big|\big|_1 + \big|\big| P(I_2) - I_2 \big|\big|_1 + \big|\big| P(I_1) - P(I_2) \big|\big|_1 \\
& \le 2 \eps + 2 \eps + \eps = 5 \eps.
\end{align*}
Furthermore, $P(\mic)$ is bounded; indeed, for any $I \in \mic$,
\begin{align*}
|| P(I) ||_1 &= \sum^N_{i=1} \int^{i \eps}_{(i-1)\eps} \big|\big(P(I)\big)(x) \big| d\Fx(x)
=\sum^N_{i=1} \int^{i \eps}_{(i-1)\eps} \frac{1}{\eps} \int^{i \eps}_{(i-1)\eps} I(y) \, dy \, d\Fx(x) \\
&=\sum^N_{i=1} \frac{\P(X \in Q_i)}{\eps} \int^{i \eps}_{(i-1)\eps} I(y) \, dy
\le \sum^N_{i=1} \frac{1}{\eps} \int^{i \eps}_{(i-1)\eps} y \, dy
= \frac{1}{2} \, N^2 \eps < \infty.
\end{align*}
Also, because $P(\mic)$ is finite dimensional, it follows that $P(\mic)$ is totally bounded.  Finally, Lemma 1 in Hanche-Olsen and Holden \cite{HH2010} implies that $\mic$ is totally bounded.
\end{proof}




\begin{proposition}\label{prop:exist}
There exists $\Is \in \mic$ such that
\begin{equation}\label{eq:exist}
\sup \limits_{I \in \mic}  \rho_b \big(u(w - X + I(X) - \pi_I) \big) =  \rho_b \big(u(w - X + \Is(X) - \pi_{\Is}) \big).
\end{equation}
\end{proposition}

\begin{proof}
From Lemma \ref{lem:mic_compact}, we know that $\mic$ is compact under the $|| \cdot ||_1$-norm.  Also, $\rho_b \big(u(w - X + I(X) - \pi_I) \big)$ is continuous in $I$ with respect to this metric and bounded above by $\rho_b(u(w)) = b(1) u(w) < \infty$.  Thus, among other things, the supremum in \eqref{eq:max} is finite.

Now, define a sequence in $\mic$ as follows: for $n \in \N$, there exists $I_n \in \mic$ such that
\begin{equation}\label{eq:In}
\rho_b \big(u(w - X + I_n(X) - \pi_{I_n}) \big) + \dfrac{1}{n} \ge \sup \limits_{I \in \mic} \rho_b \big(u(w - X + I(X) - \pi_I) \big).
\end{equation}
Because $\mic$ is a compact metric space, the sequence $\{ I_n \}$ has a subsequence $\{ I_{n_k} \}$ that converges to a function $\Is$ in $\mic$.  Also, convexity of $b$ and concavity of $u$ guarantee that $\rho_b \big(u(w - X + I(X) - \pi_I) \big)$ is continuous in $I$, which implies
\[
\lim_{n_k \to \infty} \rho_b \big(u(w - X + I_{n_k}(X) - \pi_{I_{n_k}}) \big) = \rho_b \big(u(w - X + \Is(X) - \pi_{\Is}) \big).
\]
This limit and inequality \eqref{eq:In} imply \eqref{eq:exist}.
\end{proof}

\begin{remark}
The existence result of Proposition {\rm \ref{prop:exist}} slightly extends Part $1$ of Lemma $2.1$ of Chi and Zhuang {\rm \cite{CZ2020}}.  In their proof, Chi and Zhuang {\rm \cite{CZ2020}} use the $L^\infty$ norm because they assume their random variables are bounded, while we use the $L^1(\Fx)$ norm.  Note that any random variable in $L^\infty$ is automatically in $L^1(\Fx)$.  \qed
\end{remark}

In the next section, we solve the optimization problem in \eqref{eq:max}, and we need the following results, so we present them here.  Because $I \in \mic$ is Lipschitz continuous, there exists a function $I'$ such that
\begin{equation}\label{eq:I_deriv}
I(x) = \int_0^x I'(t) \, dt = \int_0^\infty I'(t) \, \id_{\{x > t\}} \, dt,
\end{equation}
Also, we can rewrite $\pi = \pi(I(X))$, for any $I \in \mic$, via the following sequence of equalities, in which we define $S_{Y}^{-1}$ by
\[
S_{Y}^{-1}(p) = \inf\big\{ t \in \R: S_{Y}(t) \le p \big\},
\]
for $0 \le p \le 1$, and in which we use $S_{I(X)}^{-1} = I(\Sx^{-1})$ (except at a countable number of points) from Proposition 4.1 of Denneberg \cite{D1994}:
\begin{align}\label{eq:pi_I}
\pi(I(X)) &= \int_0^\infty \big[(1 + \tet)S_{I(X)}(t) + k(S_{I(X)}(t)) \big] dt  \notag \\
&= \int_0^1 S_{I(X)}^{-1}(p) \, d((1 + \tet)p +  k(p)) \notag \\
&= \int_0^1 I(S_{X}^{-1}(p)) \, d((1 + \tet)p +  k(p)) \notag \\
&= - \int_0^\infty I(t) \, d\big((1 + \tet)\Sx(t) + k(\Sx(t))\big) \notag \\
&= \int_0^\infty I'(t) \big((1 + \tet)\Sx(t) + k(\Sx(t))\big) \, dt.
\end{align}
We can similarly rewrite $\rho_b\big( u(w - X + I(X) - \pi_I) \big)$ in \eqref{eq:RDEU}, as we prove in the following lemma.

\begin{lemma}\label{lem:RDEU2}
The following identity holds:
\begin{align}\label{eq:RDEU2}
\rho_b\big( u(w - X + I(X) - \pi_I) \big) &= \int_0^1 u(w - R(S_{X}^{-1}(1 - p)) - \pi_I) \, db(p) \notag \\
&= \int_0^\infty u(w - R(x) - \pi_I) \, db(\Fx(x)),
\end{align}
in which $R(x) = x - I(x)$.
\end{lemma}

\begin{proof}
From pages 61f of Denneberg \cite{D1994}, we know that the expression for $\rho_b\big( u(w - X + I(X) - \pi_I) \big)$ in \eqref{eq:RDEU} equals
\begin{equation}\label{eq:RDEU2.5}
\int_0^{b(1)} \check S_{b \circ \P, \, u(w - R(X) - \pi_I)}(p') \, dp',
\end{equation}
in which $\check S_{\mu, f(X)}$ denotes any (pseudo-)inverse of $S_{\mu, f(X)}(\cdot) := \mu(f(X) > \cdot \,)$; see, page 5 of Denneberg \cite{D1994} for the definition of a (pseudo-)inverse of a non-increasing function.  Now, let $p' = b(p)$ in \eqref{eq:RDEU2.5} to obtain
\begin{equation}\label{eq:RDEU3}
\int_0^1 \check S_{b \circ \P, \, u(w - R(X) - \pi_I)}(b(p)) \, db(p).
\end{equation}
Next, from Proposition 4.1 of Denneberg \cite{D1994}, because $u$ is increasing, \eqref{eq:RDEU3} equals
\begin{equation}\label{eq:RDEU4}
\int_0^1 u\big(\check S_{b \circ \P, \, w - R(X) - \pi_I}(b(p)) \big) \, db(p),
\end{equation}
and because $b$ is increasing, \eqref{eq:RDEU4} simplifies to
\begin{equation}\label{eq:RDEU5}
\int_0^1 u\big(\check S_{\P, \, w - R(X) - \pi_I}(p) \big) db(p).
\end{equation}
Moreover, because $R(\cdot)$ is a non-decreasing function, we can rewrite \eqref{eq:RDEU5} as
\[
\int_0^1 u(w - R(S_X^{-1}(1 - p)) - \pi_I) \, db(p),
\]
which is the first expression for $\rho_b\big( u(w - X + I(X) - \pi_I) \big)$ in \eqref{eq:RDEU2}.  Recall that we can use any (pseudo-)inverse of $S_{R(X)}$ in this integral, including $R \circ S_X^{-1}$, which equals $S_{R(X)}^{-1}$, except at countably many points.  Finally, if we let $x = S_X^{-1}(1 - p)$, then we get the second expression in \eqref{eq:RDEU2}.
\end{proof}

We end this section with a proposition that gives two conditions under which $\Is \in \mic$ in \eqref{eq:exist} is unique, in which uniqueness means that if $I_1, I_2 \in \mic$ are optimal, then $I_1(X) = I_2(X), \, \P$-a.s.  We rely on the expression for $\rho_b\big( u(w - X + I(X) - \pi_I) \big)$ in \eqref{eq:RDEU2}.

\begin{proposition}\label{prop:unique}
Assume the utility function $u$ is strictly concave.  Then, the optimal indemnity $\Is \in \mic$ is unique if and only if at least one of the following two conditions holds:
\begin{enumerate}
\item[$1.$]  $\tet \ne 0$, in which $\tet$ equals the proportional loading factor in the distortion-deviation premium principle given in  \eqref{eq:canon}.

\item[$2.$] ${\rm ess \; inf \,}X = 0$, in which {\rm ess inf} is the $\P$-essential infimum of $X$.
\end{enumerate}
\end{proposition}

\begin{proof}  {\bf Proof of the {\it if} statement:}
Let $\mM$ denote $\max \limits_{I \in \mic}  \rho_b \big(u(w - X + I(X) - \pi_I) \big)$. Suppose $I_1, I_2 \in \mic$ are such that, for $i = 1, 2$,
\[
\rho_b \big(u(w - X + I_i(X) - \pi_i) \big) = \mM,
\]
in which $\pi_i = \pi(I_i(X))$.  Then, for $\la \in [0, 1]$, we know $I_\la := \la I_1 + (1 - \la) I_2 \in \mic$, and from Property 3 of Lemma \ref{lem:pi_prop}, we have
\[
\pi_{I_\la} \le \la \pi_1 + (1 - \la) \pi_2,
\]
which implies, because $u$ is increasing and concave,
\begin{align*}
\mM &\ge \rho_b \big(u(w - X + I_\la(X) - \pi_{I_\la}) \big) \ge \rho_b \big(u(w - X + I_\la(X) - (\la \pi_1 + (1 - \la) \pi_2)) \big) \\
&= \int_0^\infty u(w - x + \la(I_1(x) - \pi_1) + (1 - \la)(I_2(x) - \pi_2)) \, db(\Fx(x)) \\
&\ge \la \int_0^\infty u(w - x + I_1(x) - \pi_1) \, db(\Fx(x)) + (1 - \la) \int_0^\infty u(w - x + I_2(x) - \pi_2) \, db(\Fx(x)) \\
&= \mM.
\end{align*}
Thus,
\[
\rho_b \big(u(w - X + I_\la(X) - \pi_{I_\la}) \big) = \mM,
\]
for all $\la \in [0, 1]$, which implies, because $u$ is strictly concave,
\[
w - X + I_1(X) - \pi_1 = w - X + I_2(X) - \pi_2,
\]
almost surely with respect to the distorted (non-additive) measure $b \circ \P$, or equivalently, because $b$ is strictly increasing,
\begin{equation}\label{eq:I1eqI2}
I_1(X) - \pi_1 = I_2(X) - \pi_2, \qquad \P\hbox{-a.s.}
\end{equation}

Now, suppose $\tet \ne 0$, as in Condition 1; then,
\[
\int_0^1 S_{I_1(X) - \pi_1}^{-1}(p) \, d((1 + \tet)p +  k(p)) = \int_0^1 S_{I_2(X) - \pi_2}^{-1}(p) \, d((1 + \tet)p +  k(p)),
\]
which implies
\[
\int_0^1 S_{I_1(X)}^{-1}(p) \, d((1 + \tet)p +  k(p)) - (1 + \tet) \pi_1 = \int_0^1 S_{I_2(X)}^{-1}(p) \, d((1 + \tet)p +  k(p)) - (1 + \tet)\pi_2,
\]
which reduces to
\[
\tet \pi_1 = \tet \pi_2.
\]
Thus, because $\tet \ne 0$, we have $\pi_1 = \pi_2$, and \eqref{eq:I1eqI2} implies $I_1(X) = I_2(X), \, \P$-a.s.

Next, suppose ${\rm ess \; inf \,}X = 0$, as in Condition 2, that is,
\begin{equation}\label{eq:essinf}
\sup \big\{ t \in \R: \P(X < t) = 0 \big\} = 0.
\end{equation}
Define a sequence $\{ x_n: n \in \N \} \subset \R^+$ as follows: for $n \in \N$, because equation \eqref{eq:essinf} implies $\P(X < 1/n) > 0$, then there exists $x_n \in [0, 1/n)$ at which \eqref{eq:I1eqI2} holds with $X = x_n$.  Indeed, if there were {\it no} value of $x \in [0, 1/n)$ at which equation \eqref{eq:I1eqI2} holds with $X = x$, then \eqref{eq:I1eqI2} would not hold with at least probability $\P(X < 1/n) > 0$, a contradiction.  Then, we have, for $n \in \N$,
\[
\pi_1 - \pi_2 = I_1(x_n) - I_2(x_n).
\]
Because each $I_i$ is continuous with $I_i(0) = 0$, and because $\lim_{n \to \infty} x_n = 0$, we have
\[
\pi_1 - \pi_2 = \lim_{n \to \infty} \big( I_1(x_n) - I_2(x_n) \big) = I_1(0) - I_2(0) = 0,
\]
or equivalently, $\pi_1 = \pi_2$, which again implies $I_1(X) = I_2(X), \, \P$-a.s.

{\bf Proof of the {\it only if} statement:}  Suppose neither condition holds, that is, $\tet = 0$ and $a := {\rm ess \; inf \,} X > 0$.  Let $\Is$ be an optimal indemnity, which we know exists from Proposition \ref{prop:exist}.

If $\Is(a) = 0$, then define an indemnity function $\Is_a$ by
\[
\Is_a(x) = \min \big(\Is(x) + a, \, x \big), \qquad x \ge 0.
\]
Because $\Is \in \mathcal I_c$ is non-decreasing and 1-Lipschitz (specifically, $0 \le \Is(x) - \Is(y) \le x - y$ for all $0 \le y \le x$), we have, for $x \ge a$,
\[
\Is(x) + a \le  \Is(a) + x = x.
\]
Hence,
\[
\Is_a(x) = \min \big(\Is(x) + a, \, x \big) = x \id_{\{x < a\}} + (\Is(x) + a) \id_{\{x \ge a\}}, \qquad x \ge 0.
\]
Because both $\Is + a$ and $x$ are non-decreasing and 1-Lipschitz, so is their minimum.  Also, from the definition of $\Is_a$, we clearly have $0 \le \Is_a(x) \le x$ for $x \ge 0$.  We have, thus, shown $\Is_a \in \mic$.

On the other hand, $\Is_a(X) = \Is(X) + a$ ($\P$-)almost surely, because $X \ge a$ almost surely.  Because $\tet = 0$, Property 1 in Lemma \ref{lem:pi_prop} implies $\pi_{\Is_a} = \pi_{\Is} + a$.  Hence, $\Is_a(X) - \pi_{\Is_a} = \Is(X) - \pi_{\Is}$ almost surely, which implies
\[
\rho_b \big(u(w - X + \Is_a(X) - \pi_{\Is_a}) \big) = \rho_b \big(u(w - X + \Is(X) - \pi_{\Is}) \big),
\]
or equivalently, $\Is_a$ is an optimal indemnity, distinct from $\Is$ because $a > 0$.

If $\Is(a) > 0$, then define another indemnity function $\Is_a$ by
\[
\Is_a(x) = \big(\Is(x) - \Is(a) \big)_+, \qquad x \ge 0.
\]
It is straightforward to verify $\Is_a \in \mic$. Moreover, $\Is_a(X) = \Is(X) - \Is(a)$ almost surely.  Again, $\tet = 0$ implies $\pi_{\Is_a} = \pi_{\Is} - \Is(a)$.  Hence, $\Is_a(X) - \pi_{\Is_a} = \Is(X) - \pi_{\Is}$ almost surely, and $\Is_a$ is an optimal indemnity, distinct from $\Is$ because $\Is(a) > 0$.

In either case, $\Is$ is not the unique optimal indemnity.
\end{proof}

\begin{remark}
The {\bf if} statement of Proposition {\rm \ref{prop:unique}} is similar to Part $2$ of Lemma $2.1$ of Chi and Zhuang {\rm \cite{CZ2020}}, and the {\bf only if} statement is new.  Note that the proof of the {\bf only if} statement relies heavily on $I \in \mic$.

To interpret the results in Proposition {\rm \ref{prop:unique}}, assume $\theta=0$ and ${\rm ess inf} X >0$. 
Hypothetically, imagine that the condition $0\le I(x)\le  x $ is not required by an indemnity $I$.
In such a setting, adjusting an optimal indemnity function $\Is\in \mathcal I_c$ by a small amount  $\eps >0$ to  a new indemnity $\Is+ \eps$ or $\Is-\eps$ does not change the buyer's RDEU, since $\theta=0$.
The condition ${\rm ess inf} X >0$ gives some room to further adjust $\Is+\eps$ or $\Is+\eps$ to a new function $\tilde I$ in $\mathcal I_c$, i.e., satisfying $0\le \tilde I(x)\le x$, which is also optimal and thus the optimal indemnity is not unique. On the other hand,  such an adjustment is prohibited if ${\rm ess inf} X =0$, and it is no longer optimal if $\theta>0$.   \qed
\end{remark}

\section{Optimal insurance}\label{sec:optins}

\subsection{Main theorem and orders between distortion functions}\label{sec:main_amb}

We begin this section with the following theorem that characterizes optimal solutions of \eqref{eq:max}, which we know exist because of Proposition \ref{prop:exist}.  Recall that Proposition \ref{prop:unique} gives necessary and sufficient conditions for the uniqueness of the optimal indemnity.

\begin{theorem}\label{thm:optins}
An indemnity $\Is \in \mic$ is an optimal solution of \eqref{eq:max} if and only if it satisfies, for $t \ge 0$,
\begin{equation}\label{eq:Is_deriv}
(\Is)'(t) =
\begin{cases}
0, &\quad L(t) < 0, \\
v(t), &\quad L(t) = 0, \\
1, &\quad L(t) > 0,
\end{cases}
\end{equation}
almost surely with respect to Lebesgue measure, in which $v$ is some function on $\R^+$ taking values in $[0, 1]$, and $L$ is defined by
\begin{equation}\label{eq:L}
L(t) = \dfrac{\int_0^\infty u'(w - \Rs(x) - \pis) \, \id_{\{x > t\}} \, db(\Fx(x))}{\int_0^\infty u'(w - \Rs(x) - \pis) db(\Fx(x))} - \big((1 + \tet)\Sx(t) + k(\Sx(t))\big),
\end{equation}
for $t \in \R^+$, with $\Rs(x) = x - \Is(x)$ and $\pis = \pi(\Is(X))$.
\end{theorem}

\begin{proof}
Suppose $\Is \in \mic$ is the buyer's optimal indemnity; then, for any $I \in \mic$, the indemnity $\Ieps$ defined by
\[
\Ieps(x) = (1 - \eps) \Is(x) + \eps I(x),
\]
for $\eps \in (0, 1)$, is also in $\mic$ because $\mic$ is closed under convex combinations.  Because the premium principle $\pi$ is convex (recall Property 3 in Lemma \ref{lem:pi_prop}), we have
\begin{align*}
\pieps &:= \pi(\Ieps(X)) = \pi((1 - \eps) \Is(X) + \eps I(X)) \\
&\le (1 - \eps) \pi(\Is(X)) + \eps \pi(I(X)) \\
&=: (1 - \eps) \pis + \eps \pi.
\end{align*}
Because $\Is$ is optimal,
\[
\rho_b\big( u(w - \Reps(X) - \pieps)\big) \le \rho_b\big( u(w - \Rs(X) - \pis) \big),
\]
in which $\Reps(x) = x - \Ieps(x)$ and $\Rs(x) = x - \Is(x)$.  Because $u$ and $b$ are increasing, this inequality and $\pieps \le (1 - \eps) \pis + \eps \pi$ imply
\[
\rho_b\big( u(w - \Reps(X) - (1 - \eps) \pis - \eps \pi) \big) \le \rho_b\big( u(w - \Rs(X) - \pis) \big),
\]
or equivalently,
\[
\rho_b\big( u((w - \Rs - \pis) - \eps( (R - \Rs) + (\pi - \pis))) \big) \le \rho_b\big( u(w - \Rs - \pis) \big),
\]
which we rewrite, by using \eqref{eq:RDEU2}, as
\begin{align}\label{eq:ineq1}
&\int_0^\infty u\big((w - \Rs(x) - \pis) - \eps( (R(x) - \Rs(x)) + (\pi - \pis))\big) db(\Fx(x)) \notag \\
&\le \int_0^\infty u(w - \Rs(x) - \pis) \, db(\Fx(x)).
\end{align}
Because $u$ is differentiable and concave and because $R - \Rs = \Is - I$, inequality \eqref{eq:ineq1} implies
\begin{align*}
&\int_0^\infty u(w - \Rs(x) - \pis) db(\Fx(x)) \\
&\quad - \eps \int_0^\infty u'\big(  (w - \Rs(x) - \pis) - \eps((R(x) - \Rs(x)) + (\pi - \pis)) \big) \\
&~~ \qquad \qquad \times \big( (\Is(x) - I(x)) + (\pi - \pis) \big) db(\Fx(x)) \\
&\le \int_0^\infty u(w - \Rs(x) - \pis) db(\Fx(x)).
\end{align*}
After we cancel the term $\int_0^\infty u(w - \Rs(x) - \pis) db(\Fx(x))$ from each side, divide by $\eps$, and take a limit as $\eps \to 0^+$, we obtain\footnote{It is legitimate to switch the order of integration and limit by the Dominated Convergence Theorem, and we can take the limit ``inside'' $u'$ because $u$ is continuously differentiable.}
\begin{equation}\label{eq:ineq2}
\int_0^\infty u'(w - \Rs(x) - \pis) \big( (\Is(x) - I(x)) + (\pi - \pis) \big) db(\Fx(x)) \ge 0.
\end{equation}
Thus, by using \eqref{eq:I_deriv} and \eqref{eq:pi_I}, inequality \eqref{eq:ineq2} becomes
\begin{align}\label{eq:ineq3}
0 &\le \int_0^\infty u'(w - \Rs(x) - \pis) \int_0^\infty \big( (\Is)'(t) - I'(t) \big) \big\{ \id_{\{x > t\}} - \big((1 + \tet)\Sx(t) + k(\Sx(t))\big)  \big\} dt \, db(\Fx(x)) \notag \\
&= \int_0^\infty \big( (\Is)'(t) - I'(t) \big) \int_0^\infty u'(w - \Rs(x) - \pis) \big\{ \id_{\{x > t\}} - \big((1 + \tet)\Sx(t) + k(\Sx(t))\big) \big\} db(\Fx(x)) \, dt \notag \\
&= \int_0^\infty u'(w - \Rs(x) - \pis) db(\Fx(x)) \cdot \int_0^\infty \big( (\Is)'(t) - I'(t) \big) L(t) dt,
\end{align}
in which $L$ is given in \eqref{eq:L}.  Because $u$ is increasing, $u'(w - \Rs - \pis) > 0$ almost surely, which implies that inequality \eqref{eq:ineq3} is equivalent to
\begin{equation}\label{eq:ineq4}
\int_0^\infty \big( (\Is)'(t) - I'(t) \big) L(t) dt \ge 0.
\end{equation}
Because $I \in \mic$ is arbitrary, we deduce that $\Is$ necessarily satisfies \eqref{eq:Is_deriv}.

Conversely, suppose $\Is$ satisfies \eqref{eq:Is_deriv}; then, the above calculations imply, for any $I \in \mic$,
\begin{align*}
&\rho_b \big( u(w - \Rs - \pis) \big) - \rho_b \big( u(w - R - \pi)\big) \\
&\ge \int_0^\infty u'(w - \Rs(x) - \pis) \big( (\Is(x) - I(x)) + (\pi - \pis) \big) db(\Fx(x))  \\
&= \int_0^\infty u'(w - \Rs(x) - \pis) db(\Fx(x)) \cdot \int_0^\infty \big( (\Is)'(t) - I'(t) \big) L(t) dt \ge 0,
\end{align*}
which implies that $\Is$ is optimal.
\end{proof}

\begin{remark}
Theorem {\rm \ref{thm:optins}} is parallel to Theorem $3.1$ in Chi and Zhuang {\rm \cite{CZ2020}}.  The expectation $\E^\P$ and the probability function $t\mapsto \mathbb{Q}(X > t)$ in their theorem correspond to expectation with respect to the $b$-distorted probability measure and the $($possibly non-monotone$)$ distorted probability function $t\mapsto (1 + \tet)\Sx(t) + k(\Sx(t))$ in \eqref{eq:L}, respectively.   \qed
\end{remark}

In corollaries of Theorem \ref{thm:optins}, we consider three orders between the distortion functions to determine optimal solutions of \eqref{eq:max}.  One can loosely think of $(1 + \tet) \Sx(x) + k(\Sx(x))$ and $1 - b(\Fx(x))$ as defining survival functions of two random variables.  We say {\it loosely} because $(1 + \tet)p + k(p)$ is not necessarily monotone.  If we  use law-invariant orders to compare these random variables, this amounts to comparing the distortions $(1 + \tet)p + k(p)$ and $1 - b(1 - p)$.  For ease of notation, we first define distortions corresponding to these two functions.

\begin{definition}{\rm
Let $\tk, \tb \in \mD$ denote the distortions defined by, respectively,
\begin{equation}\label{eq:tk}
\tk(p) = (1 + \tet)p + k(p),
\end{equation}
and
\begin{equation}\label{eq:tb}
\tb(p) = 1 - b(1 - p),
\end{equation}
for all $p \in [0, 1]$.  Note that $\tb$ is concave because $b$ is convex; also, $\tk$ is concave because $k \in \mD$ is concave.  \qed
}
\end{definition}

We, next, define orders between members of $\mD$ that correspond to the usual definitions between random variables.  We will apply these orders to compare $\tk$ and $\tb$. For an introduction to stochastic orders, we recommend Shaked and Shanthikumar \cite{SS2007}.

\begin{definition}\label{def:ambig}
{\rm
Let $j_1, j_2 \in \mD$ be two distortions.
\begin{enumerate}
\item{}  If $j_1(p) \le j_2(p)$ for all $p \in [0, 1]$, then we say that $j_1$ is less than $j_2$ in {\it first-order stochastic dominance} (FSD) and write $j_1 \preceq_{fsd} j_2$.

\item{}  If
\begin{equation}\label{eq:HR_distortion}
\dfrac{j_1(p)}{j_2(p)}
\end{equation}
is non-decreasing with respect to $p \in (0, 1)$, then we say that $j_1$ is less than $j_2$ in {\it hazard rate} (HR) order and write $j_1 \preceq_{hr} j_2$.

\item{} If
\begin{equation}\label{eq:LR_distortion}
\dfrac{j'_1(p)}{j'_2(p)}
\end{equation}
is non-decreasing with respect to $p \in (0, 1)$, then we say that $j_1$ is less than $j_2$ in {\it likelihood ratio} (LR) order and write $j_1 \preceq_{lr} j_2$.  Here, we assume $j_1$ and $j_2$ are continuously differentiable. 
\qed
\end{enumerate}
}
\end{definition}

In Chapter 1, Shaked and Shanthikumar \cite{SS2007} prove
\[
j_1 \preceq_{lr} j_2  \implies j_1 \preceq_{hr} j_2,
\]
and
\[
 j_1 \preceq_{hr} j_2 \hbox{ and } j_1(1) \le j_2(1) \implies j_1 \preceq_{fsd} j_2.
\]
Note that $j_1 \preceq_{fsd} j_2$ if and only if the ratio in \eqref{eq:HR_distortion} is uniformly bounded above by $1$.

\subsection{Optimal insurance when $\tk \preceq_{fsd} \tb$ or $\tk \preceq_{hr} \tb$}\label{sec:k_prec_b}

In this section, we prove two corollaries of Theorem \ref{thm:optins} when $\tk \preceq_{fsd} \tb$ or $\tk \preceq_{hr} \tb$.  For the first corollary, we consider a slightly weaker version of the relation $\tk \preceq_{fsd} \tb$.  In this case, full insurance is optimal, as we show in the following corollary.

\begin{corollary}\label{cor:FSD}
Full insurance is an optimal solution of \eqref{eq:max} if and only if $\tk(p) \le \tb(p)$ for all $p \in [0, \Sx(0)]$.\footnote{If $\Sx(0) = 1$, then $\tk(p) \le \tb(p)$ for all $p \in [0, \Sx(0)]$ is equivalent to $\tk \preceq_{fsd} \tb$. }
\end{corollary}

\begin{proof}
If we set $\Rs \equiv 0$ and $\pis = \pi(X)$ in the expression for $L$ in \eqref{eq:L}, then we get
\[
L_{\{\Rs \equiv 0 \}}(x) = \big(1 - b(\Fx(x)) \big) - \big((1 + \tet)\Sx(x) + k(\Sx(x))\big) = \tb(\Sx(x)) - \tk(\Sx(x)).
\]
It follows from Theorem \ref{thm:optins} that full insurance is optimal if and only if $L_{\{\Rs \equiv 0 \}}(x) \ge 0$ for all $x \ge 0$, which is equivalent to $\tk(p) \le \tb(p)$ for all $p \in [0, \Sx(0)]$.
\end{proof}

\begin{remark}
Corollary {\rm \ref{cor:FSD}} is consistent with Theorem $4.1(ii)$ of Ghossoub and He {\rm \cite{GH2020}}, who proved that if the underwriter/insurer has a linear utility function and if $\tk$ and $\tb$ are concave, then optimal insurance is full coverage $($a so-called {\rm firm-commitment contract)} if and only if $\tk \le \tb$, modulo a technical detail concerning the reservation utility of the insurer.  \qed
\end{remark}


For the second corollary, we suppose $\tk \preceq_{hr} \tb$, which implies
\begin{equation}\label{eq:hazard}
\dfrac{\tk'(\Sx(x))}{\tk(\Sx(x))} \ge \dfrac{\tb'(\Sx(x))}{\tb(\Sx(x))},
\end{equation}
an inequality between generalized hazard rate functions (modulo $X$'s probability density function, if it has one), hence, the name: hazard rate order.  
Note that inequality \eqref{eq:hazard} requires, among other things, that $\tk$ be strictly increasing on $[0, \Sx(0)]$ because $\tb$ is automatically strictly increasing in that interval.  In this case, we show that a deductible insurance is optimal.

\begin{corollary}\label{cor:HR}
If $\tk \preceq_{hr} \tb$, then there exists $d \ge 0$ such that $\Is(x) = (x - d)_+$ is an optimal solution of \eqref{eq:max}.
\end{corollary}

\begin{proof}
Consider deductible insurance with indemnity $\Id(x) = (x - d)_+$ and retention $\Rd(x) = \min(x, d)$ for some $d \ge 0$.  Let $\pid$ denote the corresponding premium for this insurance, which equals
\[
\pid = \int_d^\infty \tk(\Sx(x))dx.
\]
Define the function $J$ on $\R^+ \times \R^+$ by
\begin{equation}\label{eq:J}
J(d, t) = \dfrac{\int_{t^+}^\infty u'(w - \Rd(x) - \pid) \frac{db(\Fx(x))}{1 - b(\Fx(t))}}{\int_0^\infty u'(w - \Rd(x) - \pid) db(\Fx(x))} - \dfrac{\tk(\Sx(t))}{1 - b(\Fx(t))},
\end{equation}
in which
\begin{align}\label{eq:Jpart}
&\int_{t^+}^\infty u'(w - \Rd(x) - \pid) \frac{db(\Fx(x))}{1 - b(\Fx(t))} \notag \\
&=
\begin{cases}
\dfrac{\int_{t^+}^d u'(w - x - \pid) db(\Fx(x)) + u'(w - d - \pid)(1 - b(\Fx(d)))}{1 - b(\Fx(t))} , &\quad 0 \le t < d, \\
u'(w - d - \pid), &\quad t \ge d.
\end{cases}
\end{align}
We assert that $J(d, t)$ is a non-decreasing function of $t$; because the ratio $\tk(\Sx(t))/(1 - b(\Fx(t)))$ is non-increasing with respect to $t \ge 0$, we only need to show that the expression in \eqref{eq:Jpart} is non-decreasing with respect to $t$ when $0 \le t < d$. To show this monotonicity, first, use integration by parts to rewrite the numerator in the first line of \eqref{eq:Jpart}:
\begin{align*}
&\int_{t^+}^d u'(w - x - \pid) db(\Fx(x)) + u'(w - d - \pid)(1 - b(\Fx(d))) \\
&= \int_{t^+}^d u''(w - x - \pid) b(\Fx(x)) dx + u'(w - d - \pid) - u'(w - t - \pid)b(\Fx(t)).
\end{align*}
Then, the expression in \eqref{eq:Jpart} is non-decreasing with respect to $t$ when $0 \le t < d$ if and only if the following is non-negative:
\begin{align*}
&- (1 - b(\Fx(t))) \, u'(w - t - \pid)b'(\Fx(t)) \\
&\quad + \left\{ \int_{t^+}^d u''(w - x - \pid) b(\Fx(x)) dx + u'(w - d - \pid) - u'(w - t - \pid)b(\Fx(t)) \right\} b'(\Fx(t)) \\
&\propto \int_{t^+}^d u''(w - x - \pid) b(\Fx(x)) dx + u'(w - d - \pid) - u'(w - t - \pid) \\
&= \int_{t^+}^d u''(w - x - \pid) (b(\Fx(x)) - 1 ) dx,
\end{align*}
which is non-negative for $t < d$ because $u$ is concave and $b(p) \le 1$.  (The symbol $\propto$ means {\it non-negatively} proportional to.)  Thus, we have shown that $J(d, t)$ is non-decreasing with respect to $t$.

Now, consider $J$ evaluated at $(d, d)$ for any $d \ge 0$:
\[
J(d, d) = \dfrac{u'(w - d - \pid)}{\int_0^\infty u'(w - \Rd(x) - \pid) db(\Fx(x))} - \dfrac{\tk(\Sx(d))}{1 - b(\Fx(d))}.
\]
If $\tk(\Sx(d)) \le 1$, then we assert that $J(d, d)$ is non-decreasing with respect to $d$.  As before, because the ratio in $\tk(\Sx(d))/(1 - b(\Fx(d)))$ is non-increasing, we only need to show that
\[
\dfrac{u'(w - d - \pid)}{\int_0^\infty u'(w - \Rd(x) - \pid) db(\Fx(x))}
\]
is non-decreasing with respect to $d$.  Equivalently, we only need to show that $\kap(d)$ in non-increasing with respect to $d$ when $\tk(\Sx(d)) \le 1$, in which $\kap$ is defined by
\[
\kap(d) = \dfrac{\int_0^\infty u'(w - \Rd(x) - \pid) db(\Fx(x))}{u'(w - d - \pid)},
\]
which equals (via integration by parts)
\[
\kap(d) = 1 + \dfrac{\int_0^d u''(w - x - \pid) b(\Fx(x)) dx - u'(w - \pid)b(\Fx(0))}{u'(w - d - \pid)}.
\]
Now,
\begin{align*}
\kap'(d) &\propto \left\{ \int_0^d u'''(w - x - \pid) b(\Fx(x)) dx - u''(w - \pid)b(\Fx(0)) \right\} u'(w - d - \pid)\tk(\Sx(d)) \\
&\quad + u'(w - d - \pid) u''(w - d - \pid) b(\Fx(d)) \\
&\quad -  \left\{ \int_0^d u''(w - x - \pid) b(\Fx(x)) dx - u'(w - \pid)b(\Fx(0)) \right\} \\
&\quad \qquad \times u''(w - d - \pid) (-1 + \tk(\Sx(d))) \\
&=  \left\{ \int_0^d u''(w - x - \pid) db(\Fx(x)) - u''(w - d - \pid)b(\Fx(d)) \right\} \\
&\qquad \times u'(w - d - \pid)\tk(\Sx(d))) \\
&\quad + u'(w - d - \pid) u''(w - d - \pid) b(\Fx(d)) \\
&\quad -  \left\{ \int_0^d u'(w - x - \pid) db(\Fx(x)) - u'(w - d - \pid)b(\Fx(d)) \right\} \\
&\quad \qquad \times u''(w - d - \pid) (-1 + \tk(\Sx(d))) \\
&= u'(w - d - \pid) \tk(\Sx(d)) \int_0^d u''(w - x - \pid) db(\Fx(x)) \\
&\quad - u''(w - d - \pid) (-1 + \tk(\Sx(d))) \int_0^d u'(w - x - \pid) db(\Fx(x)) \\
&\le - u''(w - d - \pid) (-1 + \tk(\Sx(d))) \int_0^d u'(w - x - \pid) db(\Fx(x)) \\
&\le 0,
\end{align*}
in which the last inequality follows from $\tk(\Sx(d)) \le 1$.  Thus, we have shown that $J(d, d)$ is non-decreasing with respect to $d$ when $\tk(\Sx(d)) \le 1$.  Moreover, when $\tk(\Sx(d)) > 1$, we have
\[
J(d, d) \le \dfrac{1}{1 - b(\Fx(d))} - \dfrac{\tk(\Sx(d))}{1 - b(\Fx(d))} < 0.
\]

Next, define $\ds$ by
\[
\ds = \inf \big\{ d \ge 0: J(d, d) \ge 0 \big\},
\]
with $\ds = \infty$ if $J(d, d) < 0$ for all $d \ge 0$.  If $\ds < \infty$, then $J(\ds, t) \le 0$ for $t < \ds$ and $J(\ds, t) \ge 0$ for $t > \ds$, which implies that
\[
L_{\{\Is = I_{\ds}\}}(t) = J(\ds, t) (1 - b(\Fx(t)))
\]
satisfies the conditions of Theorem \ref{thm:optins}.  Thus, $I_{\ds}$ is an optimal indemnity.  If $\ds = \infty$, then
\[
L_{\{\Is \equiv 0\}}(t) = \lim_{d \to \infty} J(d, t) (1 - b(\Fx(t)) \le  \lim_{d \to \infty} J(d, d) (1 - b(\Fx(t)) \le 0,
\]
which implies that no insurance ($\ds = \infty$) is optimal.
\end{proof}

\begin{remark}
If $\tk(\Sx(0)) \le \tb(\Sx(0))$ in Corollary \ref{cor:HR}, then $\ds = 0$, which means full insurance is optimal.  Note that $\tk(\Sx(0)) \le \tb(\Sx(0))$ and $\tk \preceq_{hr} \tb$ imply $\tk(p) \le \tb(p)$ for all $p \in [0, \Sx(0)]$, and Corollary \ref{cor:FSD} implies full insurance is optimal.  Thus, Corollaries \ref{cor:FSD} and \ref{cor:HR} are consistent.  \qed
\end{remark}

\begin{remark}
Corollary {\rm \ref{cor:HR}} extends the theorem of Arrow {\cite{A1963}}, which we stated in the Introduction, to the case for which premium is computed according to a distorted-deviation premium principle.   By contrast, Arrow assumed that the premium was an increasing function of $\E I(X)$.  Furthermore, $b(p) = p$ in Arrow's work.  \qed 
\end{remark}

\subsection{Optimal insurance when $\tb \preceq_{hr} \tk$ or $\tb \preceq_{lr} \tk$}\label{sec:b_prec_k}

In Section \ref{sec:k_prec_b}, we assume that $\tk$ precedes $\tb$, and in this section, we consider the case for which $\tb$ precedes $\tk$.  For the next corollary, we assume that $\tb \preceq_{hr} \tk$ under the special case of $u$ representing risk-neutral preferences, that is, $u(y) = y$ for all $y \in \R$.  In this case, insurance with a maximum limit is an optimal solution of \eqref{eq:max}, which we prove in the following.

\begin{corollary}\label{cor:bhrk}
If $\tb \preceq_{hr} \tk$, and if $u$ is the identity function on $\R$, then there exists $m \in [0, \infty]$ such that $\Is(x) = \min(x, m)$ is an optimal solution of \eqref{eq:max}.
\end{corollary}

\begin{proof}
Because $u(y) = y$ for all $y \in \R$, then
\[
\dfrac{\int_0^\infty u'(w - \Rs(x) - \pis) \, \id_{\{x > t\}} \, db(\Fx(x))}{\int_0^\infty u'(w - \Rs(x) - \pis) db(\Fx(x))} = \int_t^\infty  db(\Fx(x)) = 1 - b(\Fx(t)),
\]
independent of $\Is \in \mic$. Thus, we have
\[
\dfrac{L(t)}{1 - b(\Fx(t))} = 1 - \dfrac{\tk(\Sx(t))}{1 - b(\Fx(t))} = 1 - \dfrac{\tk(\Sx(t))}{\tb(\Sx(t))},
\]
and the relation $\tb \preceq_{hr} \tk$ implies this expression is non-increasing with respect to $t \ge 0$.  Define $m$ by
\[
m = \inf \big\{ t \ge 0: L(t) \le 0 \big\},
\]
in which $m = \infty$ if $L(t) > 0$ for all $t \ge 0$.  Note that $L(t) \ge 0$ for all $t < m$ and $L(t) \le 0$ for all $t > m$, which implies $\Is(x) = \min(x, m)$ is an optimal solution of \eqref{eq:max}.
\end{proof}

As a specific application of Corollary \ref{cor:bhrk}, we present the following example.  

\begin{example}
{\rm Suppose $\tk(p) = (1 + \tet)p + \alp (p \wedge (1-p))$ and $\tb(p) = p$ for $p \in [0, 1]$. Recall in Example \ref{ex:1}, the distortion $h(p) = p \wedge (1-p)$ yields the mean-median deviation measure. We see in this case $\tb \preceq_{hr}\tk$, because
\[
\dfrac{\tk(p)}{\tb(p)} = (1+\tet) + \alp (1 \wedge ({1}/{p} - 1))
\]
is non-increasing with respect to $p \in (0, 1)$. According to Corollary \ref{cor:bhrk}, if $u$ is the identity function on $\R$, then $I^*(x) = \min(x,m)$ is an optimal solution. Furthermore, by the proof of Corollary \ref{cor:bhrk}, one can derive $m=0$. Thus, no insurance is optimal, that is, $I^* \equiv 0$.
\qed
}
\end{example}

For the remainder of the paper, we impose the following conditions on our model.

\begin{assumption}\label{assum}
\begin{itemize}
\item[$(a)$] $u$ is strictly concave.

\item[$(b)$] $X$'s distribution has a point mass $1 - q \in [0, 1)$ at $0$ with a continuous density $\fx(x)$ for $x > 0$, that is, for $x \ge 0$,
\[
\Fx(x) = \P(X \le x) = (1 - q) + \int_0^x \fx(t) dt,
\]
in which $\int_0^\infty \fx(t) dt = q$.

\item[$(c)$] There exists $M > 0$ such that $\fx(x) > 0$ for all $x \in (0, M)$ and $\fx(x) = 0$ for all $x > M;$ $M$ might equal infinity.  \qed
\end{itemize}
\end{assumption}

Under Assumption \ref{assum}, Proposition \ref{prop:unique} implies that that $\Is$ is unique on $[0, M)$.  Use $X$'s model in Assumption \ref{assum}(b) to rewrite $L(t)$ and to compute $L'(t)$: for $t \ge 0$,
\begin{equation}\label{eq:L2}
L(t) = \dfrac{\int_t^M u'(w - \Rs(x) - \pis) \tb'(\Sx(x)) \fx(x) dx}{u'(w - \pis) b(1 - q) + \int_0^M u'(w - \Rs(x) - \pis) \tb'(\Sx(x)) \fx(x) dx} - \tk(\Sx(t)),
\end{equation}
and
\begin{equation}\label{eq:Lprime}
L'(t) = \fx(t) \tb'(\Sx(t)) \phi(t).
\end{equation}
in which $\phi$ is given by
\begin{equation}\label{eq:phi}
\phi(t) = \dfrac{\tk'(\Sx(t))}{\tb'(\Sx(t))} - \dfrac{u'(w - \Rs(t) - \pis)}{u'(w - \pis) b(1 - q) + \int_0^M u'(w - \Rs(x) - \pis) \tb'(\Sx(x)) \fx(x) dx}.
\end{equation}

\medskip

For the fourth corollary, we assume that $\tb \preceq_{lr} \tk$, with $\tet \ge 0$.  In this case, optimal insurance has a deductible $d$, which could equal $0$ or $\infty$, with partial insurance above that.

\begin{corollary}\label{cor:LR}
Suppose Assumption {\rm \ref{assum}} holds.  If $\tet \ge 0$ and $\tb \preceq_{lr} \tk$, then $L(t) \le 0$ for all $t \ge 0$, in which $L$ is given in \eqref{eq:L2}.
\end{corollary}

\begin{proof}
The ratio $\frac{\tb'(p)}{\tk'(p)}$ is non-decreasing with respect to $p$ if and only if $\frac{\tk'(p)}{\tb'(p)}$ is non-increasing with respect to $p$, and the later is automatically non-increasing at any point for which $\tk'(p) \le 0$, which at most occurs in a left-neighborhood of $1$.

On the contrary, suppose $L(x_0) > 0$ for some $x_0 \in (0, M)$; then, define
\[
x_1 = \inf\big\{ x \in [0, x_0]: L(y) > 0, \, \forall y \in (x, x_0] \big\},
\]
and
\[
x_2 = \sup\big\{ x \in (x_0, M): L(y) > 0, \, \forall y \in [x_0, x) \big\}.
\]
If $x_1 > 0$, then $L'(x_1-) \ge 0$, which is equivalent to
\begin{align}\label{eq:ineq_L'}
\dfrac{\tk'(\Sx(x_1))}{\tb'(\Sx(x_1))} - \dfrac{u'(w - \Rs(x_1) - \pis)}{u'(w - \pis)b(1 - q) + \int_0^M u'(w - \Rs(x) - \pis) \tb'(\Sx(x)) \fx(x) dx} \ge 0.
\end{align}
Because $L$ is continuous on $\R^+$, the interval $(x_1, x_2)$ is non-empty, and for all $t \in (x_1, x_2)$, we have $L(t) > 0$, which implies $(\Is)'(t) = 1$, or $\Rs(t) = \Rs(x_1)$.  Thus, for all $t \in (x_1, x_2)$, we have
\begin{align*}
&L'(t) = \fx(t) \tb'(\Sx(t)) \left(\dfrac{\tk'(\Sx(t))}{\tb'(\Sx(t))} - \dfrac{u'(w - \Rs(t) - \pis)}{u'(w - \pis) b(1 - q) + \int_0^M u'(w - \Rs(x) - \pis) \tb'(\Sx(x)) \fx(x) dx} \right)\\
&= \fx(t) \tb'(\Sx(t)) \left(\dfrac{\tk'(\Sx(t))}{\tb'(\Sx(t))} - \dfrac{u'(w - \Rs(x_1) - \pis)}{u'(w - \pis) b(1 - q) + \int_0^M u'(w - \Rs(x) - \pis) \tb'(\Sx(x)) \fx(x) dx} \right) \\
&\ge \fx(t) \tb'(\Sx(t)) \left(\dfrac{\tk'(\Sx(x_1))}{\tb'(\Sx(x_1))} - \dfrac{u'(w - \Rs(x_1) - \pis)}{u'(w - \pis) b(1 - q) + \int_0^M u'(w - \Rs(x) - \pis) \tb'(\Sx(x)) \fx(x) dx} \right) \\
&\ge 0,
\end{align*}
in which the first inequality follows from $\tb \preceq_{lr} \tk$; the second, from \eqref{eq:ineq_L'}.  It follows that $L(x_2) > 0$, so $x_2$ is not the maximal point $x$ for which $L(y) > 0$ for all $y \in [x_0, x)$. By iterative analysis, we must have $L(t) > 0$ for all $t \in (x_1, M)$, which implies $\Rs(x) = \Rs(x_1)$ for all $x \in [x_1, M)$ and
\begin{align*}
0 &< \dfrac{\int_t^M u'(w - \Rs(x) - \pis) \tb'(\Sx(x)) \fx(x) dx}{u'(w - \pis) b(1 - q) + \int_0^M u'(w - \Rs(x) - \pis)  \tb'(\Sx(x)) \fx(x) dx} - \tk(\Sx(t)) \\
&= \dfrac{u'(w - \Rs(x_1) - \pis)  \tb(\Sx(t))}{u'(w - \pis) b(1 - q) + \int_0^M u'(w - \Rs(x) - \pis) \tb'(\Sx(x)) \fx(x) dx} - \tk(\Sx(t)-) \\
&= \tb(\Sx(t)) \left( \dfrac{u'(w - \Rs(x_1) - \pis)}{u'(w - \pis) b(1 - q) + \int_0^M u'(w - \Rs(x) - \pis) \tb'(\Sx(x)) \fx(x) dx} - \dfrac{\tk(\Sx(t))}{\tb(\Sx(t))} \right)\\
&\le \tb(\Sx(t)) \left( \dfrac{u'(w - \Rs(x_1) - \pis)}{u'(w - \pis) b(1 - q) + \int_0^M u'(w - \Rs(x) - \pis) \tb'(\Sx(x)) \fx(x) dx} - \dfrac{\tk'(\Sx(t))}{\tb'(\Sx(t))} \right) \\
&\le \tb(\Sx(t)) \left( \dfrac{u'(w - \Rs(x_1) - \pis)}{u'(w - \pis) b(1 - q) + \int_0^M u'(w - \Rs(x) - \pis) \tb'(\Sx(x)) \fx(x) dx} - \dfrac{\tk'(\Sx(x_1))}{\tb'(\Sx(x_1))} \right),
\end{align*}
in which the second inequality follows from $\frac{\tk(\Sx(t))}{\tb(\Sx(t))} \ge \frac{\tk'(\Sx(t))}{\tb'(\Sx(t))}$ when the latter is non-decreasing with respect to $t$,\footnote{Here is a short proof of that fact: Suppose $\frac{\tk'(\Sx(t))}{\tb'(\Sx(t))}$ is non-decreasing with respect to $t$, or equivalently, $\frac{\tk'(p)}{\tb'(p)}$ is non-increasing with respect to $p$.  Then, for $0 \le \hat p \le p \le 1$, $\tb' > 0$ and $\frac{\tk'(\hat p)}{\tb'(\hat p)} \ge \frac{\tk'(p-)}{\tb'(p)}$ implies
\begin{align*}
&\tk'(\hat p) \tb'(p) \ge \tk'(p) \tb'(\hat p) \implies \int_0^p \tk'(\hat p) d\hat p \cdot \tb'(p) \ge \int_0^p \tb'(\hat p) d\hat p \cdot \tk'(p) \\
&\implies \tk(p) \tb'(p) \ge \tb(p) \tk'(p) \implies \frac{{\tk(p)}}{\tb(p)} \ge \frac{\tk'(p)}{\tb'(p)}.
\end{align*}} and the third inequality follows from $\tb \preceq_{lr} \tk$, which is equivalent to the ratio $\frac{\tk'(\Sx(t))}{\tb'(\Sx(t))}$ non-decreasing with respect to $t$.  The positivity of the expression in square brackets contradicts inequality \eqref{eq:ineq_L'}, from which we deduce that $x_1 = 0$ if $L(x_0) > 0$ for some $x_0 \in (0, M)$.

Now, consider $x_1 = 0$.  If $q = 1$ and $\tet > 0$, then $L(0) = -\tet < 0$, which contradicts $x_1 = 0$ because $L$ is continuous.  If $q = 1$ and $\tet = 0$, then $L(0) = 0$, and $x_1 = 0$ implies $L'(0) \ge 0$ because $L > 0$ in a right-neighborhood of $0$. Otherwise, if $q < 1$ and $x_1 = 0$, then we have
\begin{align*}
0 &\le L(0) = \dfrac{\int_0^M u'(w - \Rs(x) - \pis) \tb'(\Sx(x)) \fx(x) dx}{u'(w - \pis) b(1 - q) + \int_0^M u'(w - \Rs(x) - \pis) \tb'(\Sx(x)) \fx(x) dx} - \tk(q) \\
&= 1 - \tk(q) - \dfrac{u'(w - \pis)b(1 - q)}{u'(w - \pis)b(1 - q) + \int_0^M u'(w - \Rs(x) - \pis) \tb'(\Sx(x)) \fx(x) dx},
\end{align*}
which implies
\begin{align*}
L'(0) &= \fx(0) \tb'(q) \left( \dfrac{\tk'(q)}{\tb'(q)} - \dfrac{u'(w - \pis)}{u'(w - \pis) b(1 - q) + \int_0^M u'(w - \Rs(x) - \pis) \tb'(\Sx(x)) \fx(x) dx} \right) \\
&\ge \fx(0) \tb'(q) \left( \dfrac{\tk'(q)}{\tb'(q)} - \dfrac{1 - \tk(q)}{b(1 - q)} \right) \\
&= \fx(0) \tb'(q) \left( \dfrac{\tk'(q)}{\tb'(q)} - \dfrac{1 - \tk(q)}{1 - \tb(q)} \right) \\
&= \fx(0) \tb'(q) \left( \dfrac{\tk'(q)}{\tb'(q)} - \dfrac{\tk(1) - \tk(q)}{\tb(1) - \tb(q)} + \dfrac{\tet}{b(1 - q)} \right) \\
& \ge \fx(0) \tb'(q) \, \dfrac{\tet}{b(1 - q)} \ge 0,
\end{align*}
in which the second inequality follows from $\tb \preceq_{lr} \tk$ via a proof similar to the one in the most recent footnote.  Thus, for all $q \in (0, 1]$, we have $L'(0) \ge 0$, and the contradiction we obtained in the case for which $x_1 > 0$ also occurs when $x_1 = 0$.  It follows that $L \le 0$ on $[0, M)$.
\end{proof}

\begin{remark}\label{rem:Y1999}
The model in Young {\rm \cite{Y1999}} satisfies the conditions in Corollary {\rm \ref{cor:LR}} because, in that paper,  $b = \tb$ is the identity, and $\tk$ is an increasing, concave distortion with $\tk(0) = 0$ and $\tk(1) = 1$, which implies $\tet = 0$ and $\tb \preceq_{lr} \tk$. \qed
\end{remark}

Consider the condition $L(x) = 0$.  If $L(x) = 0$ holds on a non-empty interval, then, for all $x$ in that interval, we have
\[
\int_x^M  u'(w - \Rs(t) - \pis)  \tb'(\Sx(t)) \fx(t) dt = \tk(\Sx(x)) \Ups,
\]
in which $\Ups$ equals the constant
\begin{align}\label{eq:Ups}
\Ups = u'(w - \pis)b(1 - q) + \int_0^M u'(w - \Rs(t) - \pis) \tb'(\Sx(t)) \fx(t) dt.
\end{align}
By differentiating the above equation under the hypotheses of Corollary \ref{cor:LR}, we get
\[
- u'(w - \Rs(x) - \pis) \tb'(\Sx(x)) \fx(x) = - \fx(x) \tk'(\Sx(x)) \Ups,
\]
or equivalently,
\begin{equation}\label{eq:3.5}
u'(w - \Rs(x) - \pis) = \ell(x) \Ups,
\end{equation}
in which we define $\ell$, the analog of the likelihood ratio, by
\begin{equation}\label{eq:ell}
\ell(x) = \dfrac{\tk'(\Sx(x))}{\tb'(\Sx(x))}.
\end{equation}
Equation \eqref{eq:3.5} generalizes equation (3.5) in Young \cite{Y1999}.  Thus, we have the following corollary (of both Theorem \ref{thm:optins} and Corollary \ref{cor:LR}), and this corollary emends Theorem 3.6 of Young \cite{Y1999}.

\begin{corollary}\label{cor:Young1999}
Suppose Assumption {\rm \ref{assum}} holds.  If $\tet \ge 0$ and $\tb \preceq_{lr} \tk$, then for any $x \ge 0$, either $(\Is)'(x) = 0$ or \eqref{eq:3.5} holds.  \qed
\end{corollary}


Propositions 4.3 and 4.4 in Chi and Zhuang \cite{CZ2020} give sufficient conditions that ensure \eqref{eq:3.5} holds at most on a single interval $(a, b) \subset \R^+$.  
We state the conditions in the following corollary without proof because the proofs of Propositions 4.3 and 4.4 in Chi and Zhuang \cite{CZ2020} apply to our model.

\begin{corollary}\label{cor:Prop4.3_4.4}
Suppose Assumption {\rm \ref{assum}} holds.  Furthermore, suppose $\tet \ge 0$ and that either of following sets of conditions holds$:$
\begin{enumerate}
\item[$1.$]  $u''' \ge 0$, and $\ell$ in \eqref{eq:ell} is increasing and concave.

\item[$2.$]  $u$ exhibits hyperbolic absolute risk aversion $(HARA)$, that is, $-\frac{u''(x)}{u'(x)} = \frac{1}{ax + m}$ for some $a \ge 0$ and $m \in \R$, and $\ln(\ell)$ is increasing and concave.
\end{enumerate}
Then, $\Is$ is given by
\begin{equation}\label{eq:DIML}
\Is(x) =
\begin{cases}
0, &\quad 0 \le x \le d, \\
x - w + \pis + (u')^{-1}(\ell(x) \Ups), &\quad d < x \le m, \\
m - w + \pis + (u')^{-1}(\ell(m) \Ups), &\quad x > m,
\end{cases}
\end{equation}
for some $0 \le d \le m \le M$.  \qed
\end{corollary}

Note that $\ell$ depends both on the distortions $\tb$ and $\tk$ and on the distribution of $X$; thus, $\ell$ or $\ln(\ell)$ increasing and concave is not a distribution-free statement relating $\tb$ and $\tk$, as opposed to the three distribution-free orders in Definition \ref{def:ambig}.

We refer to an indemnity such as the one given in \eqref{eq:DIML} {\it deductible insurance with a maximum limit} (DIML). Note that we will not necessarily have $(\Is)'(x) = 1$ for $d < x < m$ as in a standard DIML policy, so we are extending the notion of DIML.

In the next section, we revisit some of the examples from Young \cite{Y1999} and emend or confirm them.

\section{Examples}\label{sec:Y1999}

Recall that the seller's distortion $\tk$ is continuous and concave.  Throughout this section, we further assume that $\tk(1) = 1 + \tet \ge 1$, that is, $\tet \ge 0$.  Initially, assume the buyer's distortion function $b$ equals the identity, but in Section \ref{sec:power}, we also consider a non-trivial, convex function $b$.  Young \cite{Y1999} assumes $\tet = 0$, $\tk$ is increasing, and $b$ equals the identity, so we generalize the examples in that paper.

Suppose the buyer's utility function $u$ is such that $u'(x) = e^{-\gam x}$ for $x \in \R$ and for some parameter $\gam > 0$; then, $\gam$ equals the (constant) absolute risk aversion, a special case of HARA preferences.   For all the examples in this section, assume the positive part of $X$ has the probability density function
\[
\fx(x) = q \la e^{-\la x},  \qquad x \ge 0,
\]
for some parameter $\la > 0$.  Then,
\begin{equation}\label{eq:L_expexp}
L(t) = \dfrac{q \la \int_{t}^\infty e^{\gam \Rs(x) - \la x} \, dx}{\xi} - \tk\big(qe^{-\la t} \big),
\end{equation}
in which $\xi$ equals the constant
\begin{equation}\label{eq:xi_expexp}
\xi = (1 - q) + q \la \int_0^\infty e^{\gam \Rs(x) - \la x} \, dx.
\end{equation}
Because the essential infimum of $X$ equals 0, Proposition \ref{prop:unique} implies that the optimal indemnity $\Is$ is unique.  Thus, throughout this section, we refer to {\it the} optimal solution determined by Theorem \ref{thm:optins}.

\subsection{Power distortion}\label{sec:power}

Suppose $\tk$ is a power distortion with $\tk(p) = (1 + \tet) p^c$ for $p \in [0, 1]$, and for some parameters $0 < c < 1$ and $\tet \ge 0$.  Also, suppose $b$ is the identity function; at the end of this section, we consider a non-trivial convex $b$.  Then, $\ell$ in \eqref{eq:ell} equals
\begin{equation}\label{eq:ell_expexp_power}
\ell(x) = c(1 + \tet)q^{c-1} \, e^{\la(1 - c)x}.
\end{equation}
The function $\ln(\ell)$ is linear with a positive slope, so $\ln(\ell)$ is increasing and concave.  Corollary \ref{cor:Prop4.3_4.4} implies that optimal insurance is the following DIML policy:
\begin{equation}\label{eq:DIML_expexp_power}
\Is(x) =
\begin{cases}
0, &\quad 0 \le x \le d, \vspace{0.5em} \\
x - \dfrac{1}{\gam} \, \ln(\ell(x) \xi), &\quad d < x \le m, \vspace{0.5em}  \\
m - \dfrac{1}{\gam} \, \ln(\ell(m) \xi), &\quad x > m,
\end{cases}
\end{equation}
for some $0 \le d \le m \le \infty$.  For $d < x \le m$, we have
\[
\Is(x) = \left(1 - \dfrac{\la(1 - c)}{\gam} \right) x - \dfrac{1}{\gam} \, \ln(c(1 + \tet) q^{c-1} \xi),
\]
a linear function of $x$.  Continuity of $\Is$ requires $\Is(d) = 0$, or equivalently,
\[
\left(1 - \dfrac{\la(1 - c)}{\gam} \right) d = \dfrac{1}{\gam} \, \ln(c(1 + \tet) q^{c-1} \xi),
\]
which implies
\[
\Is(x) = \left(1 - \dfrac{\la(1 - c)}{\gam} \right) (x - d),
\]
for $d < x \le m$, which only makes sense if the coefficient of $(x - d)$ is positive.  Indeed, if $\la(1 - c) \ge \gam$, then no insurance is optimal, as we prove in the following proposition.

\begin{proposition}\label{prop:expexp_power_no}
If $\la(1 - c) \ge \gam$ for the model in this section, then $\Is \equiv 0$.
\end{proposition}

\begin{proof}
If we show $L_{\{\Is \equiv 0\}}(t) \le 0$ for all $t \ge 0$, then the proposition follows from Theorem \ref{thm:optins}.  We compute
\[
L_{\{\Is \equiv 0\}}(t) = \dfrac{q \la e^{-(\la - \gam)t}}{\la - (1 - q)\gam} - (1 + \tet) q^c e^{-\la c t},
\]
which is less than or equal to $0$ for all $t \ge 0$ if and only if $L_{\{\Is \equiv 0\}}(0) \le 0$, which is equivalent to
\[
q^{1-c} \la \le (1 + \tet) \big(\la - (1 - q)\gam \big).
\]
Because $\tet \ge 0$, this inequality holds if it holds when $\tet = 0$, that is, if
\begin{equation}\label{ineq1}
q^{1-c} \la \le \la - (1 - q)\gam,
\end{equation}
and the right side of \eqref{ineq1} equals
\[
\la - (1 - q)\gam = (\la(1 - c) - \gam)(1 - q) + \la(c + (1 - c)q),
\]
with $(\la(1 - c) - \gam)(1 - q) \ge 0$.  So, inequality \eqref{ineq1} holds if the following stronger inequality holds:
\begin{equation}\label{eq:q_1-c}
q^{1 - c} \le c + (1 - c)q,
\end{equation}
for all $0 \le q \le 1$.  Because $0 < c < 1$, $q^{1-c}$ is a concave function of $q$, so its graph lies below its tangent lines, and $c + (1-c)q$ is its tangent line at $q = 1$.  Thus, inequality \eqref{eq:q_1-c} holds, and we have proved this proposition.
\end{proof}

Essentially, Proposition \ref{prop:expexp_power_no} says that if the coefficient of absolute risk aversion $\gam$ is small enough, then it is optimal for the individual not to buy insurance.  It is interesting that ``small enough'' only depends on $\la$ and $c$; it is independent of both the proportional risk loading $\tet$ and the probability of a positive loss $q$.

Henceforth, in this section, assume $\la(1 - c) < \gam$, which implies that the slope of $\Is$ on $(d, m)$ is strictly between $0$ and $1$.  In that case, optimal insurance is deductible insurance with a constant rate of coinsurance and no maximum limit, that is, $\ms = \infty$, as we prove in the following proposition.

\begin{proposition}\label{prop:expexp_power_yes}
If $\la(1 - c) < \gam$ for the model in this section, then
\[
\Is(x) = \left(1 - \dfrac{\la(1 - c)}{\gam} \right) (x - \ds)_+,
\]
in which $\ds \ge 0$ uniquely solves
\begin{equation}\label{eq:ds}
\begin{cases}
e^{(\gam - \la)d} \left\{ q^{1-c} e^{\la c d} - (1 + \tet) q \, \dfrac{\gam - \la(1 - c)}{\gam - \la} \right\} = c (1 + \tet) \left(1 - \dfrac{q\gam}{\gam - \la} \right), &\quad \la \ne \gam, \vspace{0.5em} \\
q^{1-c} e^{\la c d} - (1 + \tet)q \la c d = (1 + \tet) \big( c + (1 - c) q \big), &\quad \la = \gam.
\end{cases}
\end{equation}
Furthermore, $\ds > 0$ if and only if either $q < 1$ or $\tet > 0$.
\end{proposition}

\begin{proof}
We prove this proposition when $\la \ne \gam$ because the proof when $\la = \gam$ is similar.  We begin by demonstrating that \eqref{eq:ds} has a unique solution.  Let $G = G(d)$ denote the left side of \eqref{eq:ds} minus the right; thus, we wish to show that $G$ has a unique non-negative zero $\ds$.  To that end, note that
\[
G(0) = q^{1-c} - (1 + \tet) \big( c + (1 - c)q \big) \le 0,
\]
in which the inequality follows from $\tet \ge 0$ and inequality \eqref{eq:q_1-c}.  Also,
\[
\lim_{d \to \infty} G(d) = \infty,
\]
because the term $e^{(\gam - \la(1 - c))d}$ dominates $G$ for $d$ large.  Finally,
\[
G'(d) = \big(\gam - \la(1 - c) \big) e^{(\gam - \la)d} \big( q^{1-c} e^{\la c d} - (1 + \tet) q \big),
\]
which implies that, as $d$ increases from $0$ to infinity, either $(i)$ $G$ first decreases from a non-positive number and then increases to infinity or $(ii)$ $G$ increases monotonically to infinity.  In either case, $G$ has a unique non-negative zero $\ds$, and $\ds$ is strictly positive if and only if $G(0) < 0$, which is true if and only if either $q < 1$ or $\tet > 0$.

If we show that $L(t) \le 0$ for all $t \ge 0$, in which $L = L_{\{\Is(x) = \alp(x - \ds)_+\}}$ and $\alp = 1 - \la(1 - c)/\gam$, then the optimality of $\Is$ follows from Theorem \ref{thm:optins}.  For $t > \ds$,  $L(t) = 0$, so it is enough to show that $L(t) \le 0$ for all $0 \le t \le \ds$ with $L(\ds) = 0$.

First, calculate $\xi$, writing $d$ in place of $\ds$ for simplicity:
\begin{align*}\label{eq:xi_expexp_power}
\xi &= (1 - q) + q \la \int_0^d e^{\gam x} e^{- \la x} \, dx + q \la \int_d^\infty e^{\gam(1 - \alp)x + \gam \alp d} e^{- \la x} \, dx \notag \\
&= (1 - q) + q \la \int_0^d e^{(\gam - \la)x} \, dx + q \la e^{(\gam - \la)d} \int_d^\infty e^{- \la c(x - d)} \, dx \notag \\
&= (1 - q) + q \la \, \dfrac{e^{(\gam - \la)d} - 1}{\gam - \la} + \dfrac{q}{c} \, e^{(\gam - \la)d}.
\end{align*}
For $0 \le t \le d \; (= \ds)$, by using the expression for $d = \ds$ in \eqref{eq:ds} when $\la \ne \gam$, one can show that
\begin{align*}
L(t) &= \dfrac{1}{\xi} \left[\frac{q \la}{\gam - \la} \left(e^{(\gam - \la)d} - e^{(\gam - \la)t}\right) + \frac{q}{c} \, e^{(\gam - \la)d} \right] - (1 + \tet) q^c e^{-\la c t} \le 0,
\end{align*}
if and only if $H(t) \le 0$, in which $H$ is defined by
\[
H(t) = e^{(\gam - \la)d} \left(1 - e^{\la c(d - t)} \right) - \dfrac{\la c}{\gam - \la} \left( e^{(\gam - \la)t} - e^{(\gam - \la)d} \right).
\]
Note that $H(d) = 0$, and $H'(t)$ is positively proportional to
\[
e^{(\gam - \la(1 - c))d} - e^{(\gam - \la(1 - c))t},
\]
which is non-negative for $0 \le t \le d$.  Thus, $L(t) \le 0$ for $0 \le t \le \ds$ with $L(\ds) = 0$, and we have proved this proposition.
\end{proof}

Now, suppose $b(p) = 1 - (1-p)^a$, in which $c < a < 1$, then 
\[
b(\Fx(x)) = 1 - \big(q e^{-\la x}\big)^a = 1 - q^a e^{-\la a x},
\]
so we essentially replace $X$ with a different mixture $X'$ of a point mass at zero and an exponential random variable such that $q' = q^a$ and $\la' = \la a$.  Also, we replace $\tk$ with a different power distortion $\hat k = (1 + \tet) p^{c'}$ such that $c'$ solves
\[
\tk(\Sx(x)) = \hat k(S_{X'}(x)) \iff \big(q e^{-\la x} \big)^c = \big(q^a e^{-\la a x}\big)^{c'}
\]
or equivalently, $c' = c/a < 1$.   Then, Propositions \ref{prop:expexp_power_no} and \ref{prop:expexp_power_yes} hold with $b(p) = p$ and $(q, \la, c)$ replaced by $b(p) = 1 - (1-p)^a$ and $(q^a, \la a, c/a)$, respectively.

\subsection{Dual power distortion}\label{sec:dual}

Suppose $\tk$ is a dual power distortion with $\tk(p) = (1 + \tet) (1-(1-p)^c)$ for $p \in [0,1]$, and for some parameters $ c > 1$ and $\tet \ge 0$.  Also, suppose $b$ is the identity function.  Then, $\ell$ in \eqref{eq:ell} equals
\begin{equation}\label{eq:ell_exp_dpower}
\ell(x) = c(1 + \tet) \big(1 - qe^{-\la x} \big)^{c-1},
\end{equation}
and
\begin{align*}
\ln \ell(x) &= \ln c + \ln(1 + \tet) + (c-1) \ln \big(1-q e^{- \la x}\big), \\
(\ln \ell(x))' &= \dfrac{(c-1)q \la  e^{- \la x}}{1-q e^{- \la x}} > 0,\\
(\ln \ell(x))'' &= - \, \dfrac{(c-1)q \la^2e^{- \la x} }{(1-q e^{- \la x})^2} < 0,
\end{align*}
that is, $\ln (\ell)$ is increasing and concave.  Corollary \ref{cor:Prop4.3_4.4}, then, implies that optimal insurance is a DIML policy, as in \eqref{eq:DIML_expexp_power}, for some $0 \le d \le m \le \infty$.  For $d < x \le m$, we have
\begin{align*}
\Is(x) & = x - \dfrac{1}{\gam} \, \ln(\ell(x) \xi) \\
       & = x - \dfrac{1}{\gam} \, \ln \big(c(1 + \tet)(1 - qe^{-\la x})^{c-1} \xi \big) \\
       & = x - \dfrac{c - 1}{\gam} \, \ln \big(1 - qe^{-\la x} \big) - \dfrac{1}{\gam} \, \ln(c(1 + \tet) \xi)
\end{align*}
Continuity of $\Is$ at $x = d$ requires $\Is(d) = 0$, that is,
\[
d - \dfrac{c - 1}{\gam} \, \ln \big(1 - qe^{-\la d} \big) = \dfrac{1}{\gam} \, \ln(c(1 + \tet) \xi)
\]
which implies
\[
\Is(x) = x - d - \dfrac{c - 1}{\gam} \ln \bigg( \dfrac{1 - qe^{-\la x}} {1 - qe^{-\la d}}\bigg),
\]
for $d < x \le m$.

As for the power-distortion model in Section \ref{sec:power}, we consider when no insurance or deductible insurance might be optimal in the following two propositions.

\begin{proposition}\label{prop:expexp_dpower_no}
$I \equiv 0$ is never optimal for the model in this section.
\end{proposition}

\begin{proof}
First, if $\la \le \gam$, then the marginal utility $\xi$ is infinite when we evaluate it at $I \equiv 0$, that is, it is optimal to increase coverage above $I \equiv 0$.  Thus, if $\la \le \gam$, then no insurance cannot be optimal.

Second, if $\la > \gam$, then $L_{\{I \equiv 0\}}$ equals
\begin{align*}
L_{\{I \equiv 0\}}(t) &= \dfrac{q \la e^{-(\la - \gam)t}}{\la - (1 - q)\gam} - (1 + \tet)\left\{1 - (1- q e^{-\la t})^c\right\} \\
&= \dfrac{e^{-(\la - \gam)t}}{\la - (1- q)\gam} \left[ q \la - (1 + \tet)(\la - (1 - q)\gam) e^{(\la - \gam)t}\left\{1 - (1- q e^{-\la t})^c\right\} \right] \\
&= \dfrac{e^{-(\la - \gam)t}}{\la - (1- q)\gam} \, \big[ q \la - (1 + \tet)(\la - (1 - q)\gam) f(t) \big],
\end{align*}
in which $f$ is defined by
\[
f(t) = e^{(\la - \gam)t}\left\{1 - (1- q e^{-\la t})^c\right\} > 0.
\]
From Theorem \ref{thm:optins}, we know that no insurance is optimal if and only if $L_{\{I \equiv 0\}}(t) \le 0$ for all $t \ge 0$.  However,
\begin{align*}
\lim_{t \to \infty} f(t) &= \lim_{t \to \infty} \dfrac{1 - (1- q e^{-\la t})^c}{e^{-(\la - \gam)t}} = \lim_{t \to \infty} \dfrac{q c \la e^{-\la t}(1- q e^{-\la t})^{c-1}}{(\la - \gam)e^{-(\la - \gam)t}} \\
&= \lim_{t \to \infty} \dfrac{q c \la(1- q e^{-\la t})^{c-1}}{(\la - \gam)e^{\gam t}} = 0,
\end{align*}
which implies that, for $t$ large enough, $L_{\{I \equiv 0\}}(t) > 0$.  Thus, if $\la > \gam$, then no insurance cannot be optimal.
\end{proof}

If the risk aversion parameter $\gam$ is large enough, then optimal insurance has no maximum limit, as we show in the following proposition.

\begin{proposition}\label{prop:expexp_dpower_yes}
If $\gam (1 - q) > q \la (c - 1)$ for the model in this section, then
\begin{align}\label{eq:dIs}
\Is(x) = 
\begin{cases}
0, &\quad 0 \le x \le \ds, \\
(x - \ds) - \dfrac{c - 1}{\gam} \, \ln \dfrac{1 - q e^{- \la x}}{1 - qe^{- \la \ds}}, &\quad x > \ds,
\end{cases}
\end{align}
in which $\ds \ge 0$ uniquely solves
\begin{equation}\label{eq:dds}
\begin{cases}
e^{\gam d}\, \dfrac{(1+\tet)(1 - q e^{-\la d})^c - \tet}{(1 - q e^{- \la d})^{c - 1} } = c(1 + \tet) \left(
1 - q + \dfrac{q \la (e^{ (\gam - \la) d}-1)}{\gam - \la} \right), &\quad \la \ne \gam, \vspace{1em} \\
e^{\la d} \, \dfrac{(1+\tet)(1 - q e^{-\la d})^c - \tet}{(1 + \tet)(1 - q e^{- \la d})^{c - 1} } =  c(1 + \tet) \big(1 - q + q \la d \big), &\quad \la = \gam.
\end{cases}
\end{equation}
Furthermore, $\ds > 0$ if and only if either $q < 1$ or $\tet > 0$.
\end{proposition}

\begin{proof}
We prove this proposition when $\la \neq \gam$ because the proof when $\la = \gam$ is similar.  We begin by demonstrating that \eqref{eq:dds} has a unique solution. Let $G= G(d)$, in which
\[
G(d) = e^{\gam d} \left\{(1+\tet)(1-q e^{-\la d})^c - \tet \right\} - c(1 + \tet) (1-qe^{-\la d})^{c-1}\left(1 - q + \dfrac{q \la (e^{(\gam - \la) d}-1)}{\gam - \la}\right),
\]
which equals left side of \eqref{eq:dds} minus the right, all multiplied by $(1-qe^{-\la d})^{c-1}$.  We wish to show $G$ has a unique non-negative zero $\ds$.  To that end, note that
\[
G(0) = - \tet - (1 + \tet)(c - 1) (1 - q)^c \le 0,
\]
and
\[
\lim_{d \to \infty} G(d) = \infty,
\]
from which it follows that $G$ has at least one non-negative zero.  Let $d_0$ denote a zero of $G$.  By differentiating $G$ and using $G(d_0) = 0$, we obtain
\begin{align*}
G'(d_0) & = \gam e^{\gam d_0} \left\{(1+\tet)(1 - qe^{-\la d_0})^c - \tet\right\}  \\
& \qquad - c(c - 1)(1 + \tet)q \la e^{-\la d_0}(1 - q e^{-\la d_0})^{c - 2} \left\{1 - q + \dfrac{q \la (e^{(\gam - \la)d_0}-1)}{\gam - \la}\right\} \\
& =  \gam c(1+\tet)(1-qe^{-\la d_0})^{c-1}\left\{1 - q + \dfrac{q \la (e^{(\gam - \la) d_0}-1)}{\gam - \la}\right\}\\
& \qquad - c(c - 1)(1 + \tet)q \la e^{-\la d_0}(1 - q e^{-\la d_0})^{c - 2} \left\{1 - q + \dfrac{q \la (e^{(\gam - \la)d_0}-1)}{\gam - \la}\right\} \\
&\propto \gam (1-qe^{-\la d_0}) - q \la (c-1)e^{-\la d_0} \ge \gam(1 - q) - q\la(c - 1) > 0,
\end{align*}
in which the last inequality follows from the hypothesis of the proposition.  Therefore, $G$ has a unique non-negative zero $\ds$, and $\ds$ is strictly positive if and only if $G(0) < 0$, which is true if and only if either $q < 1$ or $\tet > 0$.

If we show $L(t) \le 0$ for all $t \ge 0$, in which $L = L_{\Is}$ and $\Is$ is defined in \eqref{eq:dIs}, then the optimality of $\Is$ follows from Theorem \ref{thm:optins}.  For $t > \ds$,  $L(t) = 0$, so it is enough to show that $L(t) \le 0$ for all $0 \le t \le \ds$ with $L(\ds) = 0$.

First, we calculate the following integral for $0 \le t \le \ds$, writing $d$ in place of $\ds$ for simplicity:
\begin{align*}
\int^{\infty}_t e^{\gam R^*(x) - \la x} d x
& = \int^{\infty}_t e^{(\gam - \la)x - \gam I^*(x)} d x \\
& = \int^{d}_t e^{(\gam - \la)x} d x + e^{\gam d} \int^{\infty}_d e^{-\la x} \left(\dfrac{1 - qe^{- \la x}}{1 - q e^{-\la d}}\right)^{c - 1} dx\\
& = \frac{1}{\gam - \la}\left(e^{(\gam - \la)d} - e^{(\gam - \la)t}\right) + \dfrac{e^{\gam d}}{q \la c}\cdot \dfrac{1 - (1 - qe^{-\la d})^c}{(1 - qe^{- \la d})^{c - 1}},
\end{align*}
from which we deduce
\begin{align*}\label{eq:xi_expexp_dpower}
\xi &= (1 - q) + q \la \int_0^{\infty} e^{\gam R^*(x) - \la x} \, dx  \notag \\
&= (1 - q) + q \la \, \dfrac{ e^{(\gam - \la)d} - 1}{\gam - \la} + \dfrac{e^{\gam d}}{c}\cdot \dfrac{1 - (1 - qe^{-\la d})^c}{(1 - qe^{- \la d})^{c - 1}}.
\end{align*}
Then, $L(t)$, for $0 \le t \le d \; (= \ds)$, equals
\begin{align*}
L(t) &= \dfrac{1}{\xi} \left[\frac{q \la}{\gam - \la} \left(e^{(\gam - \la)d} - e^{(\gam - \la)t}\right) + \frac{e^{\gam d}}{c}\dfrac{1 - (1 - qe^{-\la d})^c}{(1 - qe^{- \la d})^{c - 1}} \right] - (1 + \tet)\left(1 - \left( 1 - q e^{-\la t}\right)^c\right),
\end{align*}
and by using the expression for $d = \ds$ in \eqref{eq:dds}, one can show that $L(t) \le 0$ if and only if $H(t) \le 0$, in which $H$ equals
\begin{equation}\label{eq:H1_d}
H(t) = \frac{q \la}{\gam - \la} \left(e^{(\gam - \la)d} - e^{(\gam - \la)t}\right) + \frac{e^{\gam d}}{c} \cdot \dfrac{1 - (1 - qe^{-\la d})^c}{(1 - qe^{- \la d})^{c - 1}} \left\{1 - \dfrac{1 - \left( 1 - q e^{-\la t}\right)^c}{1 - \left( 1 - q e^{-\la d}\right)^c} \right\}.
\end{equation}
Note that $H(d) = 0$, and
\begin{align}\label{eq:derivH}
H'(t) &= -q \la e^{(\gam - \la)t} + q \la e^{-\la t} e^{\gam d} \left(\dfrac{1 - qe^{-\la t}}{1 - qe^{-\la d}}\right)^{c-1} \notag\\
& \propto - 1 + e^{\gam (d - t)} \left(\dfrac{1 - qe^{-\la t}}{1 - qe^{-\la d}}\right)^{c-1}.
\end{align}
Denote the right side of \eqref{eq:derivH} by $h$, then $h(d) = 0$ and
\begin{align*}
h'(t) & =  - \gam e^{\gam (d - t)} \left(\dfrac{1 - qe^{-\la t}}{1 - qe^{-\la d}}\right)^{c-1}
+ e^{\gam (d - t)}  (c - 1)q \la e^{-\la t} \, \dfrac{(1 - qe^{-\la t})^{c - 2}}{(1 - qe^{-\la d})^{c-1}}\\
& \propto - \gam (1 - qe^{-\la t}) + (c - 1)q \la e^{-\la t} \\
& \le - \gam (1 - q) + (c - 1)q \la < 0,
\end{align*}
in which the last inequality follows from the hypothesis of the proposition.  Thus, $h(t)$ is non-negative for $0 \le t \le d$, which implies $H'(t) \ge 0$ for $0 \le t \le d$.  Because $H(d) = 0$, we deduce that $H(t) \le 0$ for all $0 \le t \le d$. Hence, $L(t) \le 0$ for all $0 \le t \le \ds$ with $L(\ds) = 0$, and we have proved this proposition.
\end{proof}


\subsection{Gini deviation}\label{sec:Gini}

Suppose $\tk(p) = (1 + \tet) p + \alp (p - p^2)$ for $p \in [0,1]$, and for some parameters $\tet \ge 0$ and $\alp \ge 0$.  From Example \ref{ex:1}, recall that the term $p - p^2$ yields the Gini deviation measure. We also see that $\tk$ is concave but not increasing.  Also, suppose $b$ is the identity function. Then, $\ell$ in \eqref{eq:ell} equals
\begin{equation}\label{eq:ell_exp_gini}
\ell(x) = (1 + \tet) + \alp (1 - 2qe^{- \la x}),.
\end{equation}
By differentiating $\ell$, we obtain
\[
\ell'(x) = 2 \alp q \la e^{-\la x} > 0,
\]
and
\[
\ell''(x) = - 2 \alp q \la^2 e^{-\la x} < 0,
\]
which shows that $\ell$ is increasing and concave.  Also, $u'''(x) \ge 0$, and it follows from Corollary \ref{cor:Prop4.3_4.4} that optimal insurance is a DIML policy, as in \eqref{eq:DIML_expexp_power}, for some $0 \le d \le m \le \infty$.  For $d < x \le m$, we have
\begin{align*}
\Is(x) & = x - \dfrac{1}{\gam} \, \ln(\ell(x) \xi) \\
       & = x - \dfrac{1}{\gam} \, \ln((1 + \tet)\xi + \alp (1 - 2qe^{-\la x}) \xi)
\end{align*}
Continuity of $\Is$ at $x = d$ requires $\Is(d) = 0$, that is,
\[
d = \dfrac{1}{\gam} \, \ln((1 + \tet) \xi + \alp (1 - 2qe^{-\la d}) \xi),
\]
which implies
\[
\xi = \dfrac{e^{\gam d}}{(1 + \tet) + \alp (1 - 2qe^{-\la d})}.
\]
Thus,
\[
\Is(x) = x - d - \frac{1}{\gam} \ln \dfrac{(1 + \tet) + \alp (1 - 2qe^{-\la x})}{(1 + \tet) + \alp (1 - 2qe^{-\la d})},
\]
for $d < x \le m$.

As in the previous two sections, we consider when no insurance or deductible insurance might be optimal in the following two propositions.

\begin{proposition}
$I \equiv 0$ is never optimal for the model in this section.
\end{proposition}

\begin{proof}
First, if $\la \le \gam$, then the marginal utility $\xi$ is infinite when we evaluate it at $I \equiv 0$, that is, it is optimal to increase coverage above $I \equiv 0$.  Thus, if $\la \le \gam$, then no insurance cannot be optimal.

Second, if $\la > \gam$, then $L_{\{I \equiv 0\}}$ equals
\begin{align*}
L_{\{I \equiv 0\}}(t) & = \dfrac{q \la e^{-(\la - \gam)t}}{\la - (1 - q)\gam} - (1 + \tet)q e^{-\la t} - \alp \big(q e^{- \la t} -q^2 e^{-2 \la t} \big) \\
& = q e^{-\la t} \left\{\dfrac{ \la e^{\gam t}}{\la - (1 - q)\gam} - (1 + \tet) - \alp \big(1 -q e^{-\la t} \big)\right\}.
\end{align*}
From Theorem \ref{thm:optins}, we know that no insurance is optimal if and only if $L_{\{I \equiv 0\}}(t) \le 0$ for all $t \ge 0$.  However,
\begin{align*}
\lim_{t \to \infty} \dfrac{\la e^{\gam t}}{\la - (1 - q)\gam} - (1 + \tet) - \alp\big(1 -q e^{-\la t} \big) = \infty,
\end{align*}
which implies that, for $t$ large enough, $L_{\{I \equiv 0\}}(t) > 0$.  Thus, if $\la > \gam$, then no insurance cannot be optimal.
\end{proof}

If the risk aversion parameter $\gam$ is large enough, then optimal insurance has no maximum limit, as we show in the following proposition.

\begin{proposition}\label{prop:expexp_gini_yes}
If $\gam \big((1 + \tet) + \alp(1 - 2q) \big) > 2\alp q \la$  for the model in this section, then
\begin{align}\label{eq:giIs}
I^*(x) =
\begin{cases}
0, &\quad 0 < x \le \ds, \vspace{0.5em} \\
x - \ds - \dfrac{1}{\gam} \, \ln \dfrac{1 + \tet + \alp (1 - 2q e^{- \la x})}{1 + \tet + \alp (1 - 2q e^{- \la \ds})}, &\quad x>\ds.
\end{cases}
\end{align}
in which $\ds \ge 0$ uniquely solves
\begin{equation}\label{eq:gids}
\begin{cases}
\dfrac{e^{\gam d}( 1 - \alp q^2 e^{-2\la d})}{(1 + \tet) + \alp(1 - 2 q e^{-\la d})} = 1 + q \gam \, \dfrac{e^{(\gam - \la)d} - 1}{\gam - \la}, &\quad \la \ne \gam,   \vspace{0.5em}\\
e^{\la d} = (1 + q \la d)\big((1 + \tet) + \alp(1 - 2 q e^{-\la d})\big) + \alp q^2 e^{-\la d}, &\quad \la = \gam.
\end{cases}
\end{equation}
Furthermore, $\ds > 0$ if and only if either $\tet > 0$ or both $q < 1$ and $\alp > 0$.
\end{proposition}

\begin{proof}
We prove this proposition when $\la \neq \gam$ because the proof when $\la = \gam$ is similar.  We begin by showing that \eqref{eq:gids} has a unique solution.  Let $G = G(d)$ denote the left side of \eqref{eq:gids} minus the right, all times the denominator $(1 + \tet) + \alp(1 - 2 q e^{-\la d})$.  We wish to show that $G$ has a unique non-negative zero $\ds$.  To that end, note that
\[
G(0) = - \tet - \alp(1 - q)^2 \le 0,
\]
and
\[
\lim_{d \to \infty} G(d) = \infty,
\]
from which it follows that $G$ has at least one non-negative zero.  Let $d_0$ denote a zero of $G$.  By differentiating $G$ and by using $G(d_0) = 0$, we obtain
\begin{align*}
G'(d_0) &= \gam e^{\gam d_0} \big(1 - \alp q^2 e^{-2\la d_0} \big) + e^{\gam d_0} \cdot 2 \alp q^2 \la e^{-2 \la d_0} \\
&\quad - q \gam e^{(\gam - \la)d_0} \big((1 + \tet) + \alp(1 - 2 q e^{-\la d_0}) \big) - \left( 1 + q \gam \, \dfrac{e^{(\gam - \la)d_0} - 1}{\gam - \la} \right) 2 \alp q \la e^{-\la d_0} \\
&= \left(1 - q + \dfrac{q \la (e^{(\gam-\la)d_0}-1)}{\gam - \la}\right)\big(\gam \big((1 + \tet) + \alp(1 - 2 q e^{-\la d_0}) \big) - 2 \alp q \la e^{-\la d_0} \big) \\
&\propto \gam \big((1 + \tet) + \alp(1 - 2 q e^{-\la d_0}) \big) - 2 \alp q \la e^{-\la d_0} \\
&\ge \gam \big((1 + \tet) + \alp(1 - 2 q) \big) - 2 \alp q \la > 0,
 \end{align*}
in which the last inequality follows from the hypothesis of the proposition.  Therefore, $G$ has a unique non-negative zero $\ds$, and $\ds$ is strictly positive if and only if $G(0) < 0$, which holds if and only if either $\tet > 0$ or both $q < 1$ and $\alp > 0$.

If we show $L(t) \le 0$ for all $t \ge 0$, in which $L = L_{\Is}$ and $\Is$ is defined in \eqref{eq:giIs}, then the optimality of $\Is$ follows from Theorem \ref{thm:optins}.  For $t > \ds$,  $L(t) = 0$, so it is enough to show that $L(t) \le 0$ for all $0 \le t \le \ds$ with $L(\ds) = 0$.

First, we calculate the following integral for $0 \le t \le \ds$, writing $d$ in place of $\ds$ for simplicity:
\begin{align*}
\int^{\infty}_t e^{\gam R^*(x) - \la x} d x
& = \int^{\infty}_t e^{(\gam - \la)x - \gam I^*(x)} d x \\
& = \int^{d}_t e^{(\gam - \la)x} d x + \dfrac{e^{\gam d}}{(1 + \tet) + \alp(1 - 2 qe^{-\la d})} \int^{\infty}_d \left\{
(1 + \tet + \alp)e^{-\la x} - 2 \alp qe^{-2 \la x}\right\}dx \vspace{0.5em}\\
& = \frac{1}{\gam - \la}\left(e^{(\gam - \la)d} - e^{(\gam - \la)t}\right) + \dfrac{e^{(\gam - \la) d}}{\la} \cdot \dfrac{(1 + \tet) + \alp(1 - q e^{-\la d})}{(1 + \tet) + \alp(1 - 2 qe^{-\la d})},
\end{align*}
from which we deduce
\begin{align}\label{eq:xi_expexp_gi}
\xi &= (1 - q) + q \la \int_0^{\infty} e^{\gam R^*(x) - \la x} \, dx  \notag \\
&= (1 - q) + q \la \, \dfrac{e^{(\gam - \la)d} - 1}{\gam - \la} + q e^{(\gam - \la) d} \, \dfrac{(1 + \tet) + \alp(1 - q e^{-\la d})}{(1 + \tet) + \alp(1 - 2 qe^{-\la d})}.
\end{align}
Then, $L(t)$, for $0 \le t \le d \; (= \ds)$, equals
\begin{align*}
L(t) &= \dfrac{1}{\xi} \left[  \frac{q \la}{\gam - \la}\left(e^{(\gam - \la)d} - e^{(\gam - \la)t}\right) + q e^{(\gam - \la) d} \, \dfrac{(1 + \tet) + \alp(1 - q e^{-\la d})}{(1 + \tet) + \alp(1 - 2 qe^{-\la d})} \right] \\
& \quad - q e^{-\la t}\big( (1 + \tet) - \alp(1 - q e^{-\la t})\big),
\end{align*}
and by using the expression for $d = \ds$ in \eqref{eq:gids}, one can show $L(\ds) = 0$,
\begin{align*}
L(t) &= \dfrac{q \la}{\xi} \, \dfrac{e^{(\gam - \la)d} - e^{(\gam - \la)t}}{\gam - \la} + q e^{-\la d}\big( (1 + \tet) + \alp(1 - q e^{-\la d})\big) \\
& \quad - q e^{-\la t}\big( (1 + \tet) + \alp(1 - q e^{-\la t})\big),
\end{align*}
and
\begin{equation}\label{eq:xi_gi}
 \xi  = \dfrac{ e^{\gam d} }{(1 + \tet) + \alp(1 - 2 qe^{-\la d})} > 0.
\end{equation}
By differentiating $L$ with respect to $t$, we obtain
\begin{align}\label{eq:deH1}
L'(t) &= - \,  \dfrac{q \la}{\xi} \, e^{(\gam - \la)t} + q \la e^{-\la t}\big( (1 + \tet) + \alp(1 - 2q e^{-\la t})\big) \notag \\
&\propto - \, \dfrac{e^{\gam t}}{\xi} +  \big( (1 + \tet) + \alp(1 - 2q e^{-\la t})\big) \notag \\
&\propto \dfrac{e^{\gam d}}{(1 + \tet) + \alp(1 - 2q e^{-\la d})} - \dfrac{e^{\gam t}}{(1 + \tet) + \alp(1 - 2q e^{-\la t})},
\end{align}
which one can show is positive for $t < d$ because \eqref{eq:deH1} decreases with respect to $t$ and equals $0$ at $t = d$.  Indeed, by differentiating the expression in \eqref{eq:deH1}, we see
\begin{align*}
&- \dfrac{d}{dt} \left( \dfrac{e^{\gam t}}{(1 + \tet) + \alp(1 - 2q e^{-\la t})} \right) \propto - \gam \big((1 + \tet) + \alp(1 - 2 q e^{-\la t}) \big) + 2 \alp q \la e^{-\la t} \\
&\le - \left[ \gam \big((1 + \tet) + \alp(1 - 2 q) \big) - 2 \alp q \la \right] < 0,
 \end{align*}
in which the last inequality follows from the hypothesis of the proposition.  Thus, $L(t) \le 0$ for all $0 \le t \le \ds$ with $L(\ds) = 0$, and we have proved this proposition.
\end{proof}

\section{Conclusion}\label{sec:conclu}

In this paper, in Theorem \ref{thm:optins}, we found necessary and sufficient conditions that the optimal indemnity satisfies for an RDEU maximizer subject to a distortion-deviation premium principle with a concave distortion $\tk(p) = (1 + \tet)p + k(p)$, $p \in [0, 1]$.  We modeled the RDEU maximizer's preferences via a concave utility function $u$ and a strictly increasing, concave distortion $\tb(p) = 1 - b(1 - p)$, $p \in [0, 1]$.  We determined conditions under which optimal insurance is full insurance (Corollary \ref{cor:FSD}), deductible insurance (Corollary \ref{cor:HR}), insurance with a maximum limit (Corollary \ref{cor:bhrk}), and insurance with a possible deductible and coinsurance above the deductible (Corollaries \ref{cor:LR}, \ref{cor:Young1999}, and \ref{cor:Prop4.3_4.4}).

As we discussed in Remark \ref{rem:two}, the convexity of $b$ and  the concavity of $k$ are only used to obtain Lemma \ref{lem:comon}, and those requirements can be relaxed if we work with $\mathcal I_c$ as the set of indemnities {\it ex ante}.  Note that the distortion function $k$ is not monotone in our setting.  From a mathematical point of view, our techniques can be applied when monotonicity of $b$ is dispensed with, although such a relaxation  is less relevant  in the RDEU model.

As for future research directions, it would be of interest to study RDEU and the distortion-deviation premiums in the context of Pareto-optimal contracts for the insured and the insurer (see, for example, Cai et al.~\cite{CLW17}) and that of optimal contracts in bargaining models such as the Nash and Kalai--Smorodinsky bargaining models (see, for example, Jiang et al.~\cite{JRYH19}).

\end{document}